\documentclass{article}
\usepackage[top=30truemm,bottom=30truemm,left=30truemm,right=30truemm]{geometry}

\usepackage[T1]{fontenc}    
\usepackage{hyperref}       
\usepackage{url}            
\usepackage{booktabs}       
\usepackage{amsfonts}       
\usepackage{nicefrac}       
\usepackage{microtype}      

\usepackage{subcaption}
\usepackage{multirow}

\usepackage{wrapfig}
\usepackage[leqno]{amsmath}
\usepackage{amsthm} 
\usepackage{amssymb}

\usepackage{algorithm}
\usepackage{algorithmic}
\usepackage{graphicx,color}

\title{Practical Adversarial Combinatorial Bandit Algorithm via Compression of Decision Sets}

\author{
	Shinsaku Sakaue
	\\
	NTT Communication Science Laboratories\\
	\texttt{sakaue.shinsaku@lab.ntt.co.jp} \\
	\and
	Masakazu Ishihata\\
	Hokkaido University\\	\texttt{ishihata.masakazu@ist.hokudai.ac.jp}
    \and
    Shin-ichi Minato\\
    Hokkaido University\\	\texttt{minato@ist.hokudai.ac.jp}
}

\newtheorem{thm}{Theorem}
\newtheorem{lem}{Lemma}

\newcommand{\zdd}[1]{\mathbf{G}_{#1}}
\newcommand{\lbl}[1]{l_{#1}}
\newcommand{\child}[2]{\mathrm{c}_{#1}^{#2}}
\newcommand{\arc}[2]{a_{#1}^{#2}}
\newcommand{\Rt}[2]{\mathcal{R}_{#1,#2}}
\newcommand{\F}[1]{\mathrm{F}_{#1}}
\newcommand{\B}[1]{\mathrm{B}_{#1}}
\newcommand{\BC}[2]{\mathrm{C}_{#1,#2}}
\newcommand{\CP}[2]{\mathrm{P}_{#1,#2}}
\newcommand{\Ber}[1]{\mathrm{Ber}\!\left(#1\right)}

\def\xb{\mbox{\boldmath $x$}}

\def\zb{\mbox{\boldmath $z$}}

\def\Fl{\mathcal{F}}

\newcommand{\ffcont}[1]{f}

\def\st{$s$-$t$\ }

\newcommand{\cost}[1]{C(P)}

\newcommand{\var}[1]{\text{Var}[{#1}]}

\def\zarc{$0$-arc}
\def\oarc{$1$-arc}

\def\mab{MAB}
\def\cmab{CMAB}

\def\bdx{L}
\newcommand{\vart}[1]{\text{Var}_t[{#1}]}
\def\tr{\text{Tr}}

\def\combexp{{C{\scriptsize OMB}EXP}}
\def\comband{{C{\scriptsize OM}B{\scriptsize AND}}}
\def\combandms{{C{\scriptsize OMB}WM}}

\def\lb{\mbox{\boldmath $\ell$}}

\def\ib{\mbox{\boldmath $1$}}
\def\wb{\mbox{\boldmath $w$}}

\def\wt{{\widetilde w}}
\def\wbt{\mbox{{\boldmath $\widetilde w$}}}
\def\xbt{\mbox{{\boldmath $\widetilde x$}}}

\def\xbb{\mbox{{\boldmath $\bar x$}}}

\def\lbh{\mbox{{\boldmath $\hat \ell$}}}
\def\lh{{\hat \ell}}

\def\Xt{{X_t}}

\def\Fl{\mathcal{S}}

\newcommand{\weight}[2]{w_{{#1},{#2}}}
\newcommand{\tlweight}[2]{\widetilde{w}_{{#1},{#2}}}
\newcommand{\loss}[2]{\ell_{{#1},{#2}}}
\newcommand{\losst}[1]{\ell_{t,#1}}
\newcommand{\hloss}[2]{\lh_{{#1},{#2}}}

\def\a{\alpha}
\def\b{\beta}

\def\norm#1{\|#1\|}

\def\R{\mathbb{R}}
\def\N{\mathbb{N}}

\def\E{\mathbb{E}}

\newcounter{counter}

\let\oldref\ref
\def\ref#1{{\normalfont\oldref{#1}}}
\def\eqref#1{{\normalfont(\oldref{#1})}}

\mathcode`@="8000 
{\catcode`\@=\active\gdef@{\mkern1mu}}

\mathchardef\Gamma="7100 \mathchardef\Delta="7101
\mathchardef\Theta="7102 \mathchardef\Lambda="7103
\mathchardef\Psi="7104 \mathchardef\Pi="7105 \mathchardef\Sigma="7106
\mathchardef\Upsilon="7107 \mathchardef\Phi="7108
\mathchardef\Psi="7109 \mathchardef\Omega="710A

\newif\iffigure
\figuretrue

\begin{document}
\maketitle

\begin{abstract}

We consider the adversarial combinatorial multi-armed bandit (CMAB) problem, whose decision set can be exponentially large with respect to the number of given arms. To avoid dealing with such large decision sets directly, we propose an algorithm performed on
a {\it zero-suppressed binary decision diagram} (ZDD), which is a compressed representation of the decision set. The proposed algorithm achieves either $O(T^{2/3})$ regret with high probability or $O(\sqrt{T})$ expected regret as the {\it any-time guarantee}, where $T$ is the number of past rounds. Typically, our algorithm works efficiently for CMAB problems defined on networks. Experimental results show that our algorithm is applicable to various large adversarial CMAB instances including adaptive routing problems on real-world networks.
\end{abstract}

\section{Introduction}\label{section:introduction}

The multi-armed bandit (\mab) problem~\cite{robbins1985some} has been extensively studied as a fundamental framework for online optimization problems with partial observations.
In the \mab, a player chooses an arm (choice) from a set of possible arms.
Then, the player incurs a cost and obtains feedback according to the selected arm.
The aim of the player is to minimize the cumulative cost by exploring possible arms and exploiting those with low costs.
There have been many studies on \mab\ applications, e.g., clinical trials~\cite{kuleshov2014algorithms}, 
and recommendation systems~\cite{li2010contextual}.

\newcommand{\osp}{OSP}
\newcommand{\dst}{DST}
\newcommand{\cg}{CG}
In many real-world problems, each possible choice that the player can make is not expressed as a single arm but as a {\it super arm}, which is a set of arms that satisfies certain combinatorial constraints; 
a full set of super arms is called a {\it decision set}.
This problem is called the {\it combinatorial multi-armed bandit} (\cmab) problem, 
and 
the \cmab\ is said to be {\it adversarial} if the cost of each arm is arbitrarily changed by an adversary. 
Examples of the adversarial \cmab\ 
include various important problems on networks such as 
the {\it online shortest path} (\osp) problem~\cite{awerbuch2004adaptive,gyorgy2007online}, the {\it dynamic Steiner tree} (\dst) problem~\cite{imase1991dynamic}, and the {\it congestion game} (\cg)~\cite{rosenthal1973games}; 
although 
the original \cg\ is a resource allocation problem over multiple players, it can be formulated as an adversarial \cmab\ if a player considers the other players to be adversaries.
For instance, 
in the \osp\ on a traffic network, an arm corresponds to an edge (road) of a given network,  
a super arm is an $s$-$t$ path that connects the current point $s$ and the destination $t$,  
and the decision set is a set of all $s$-$t$ paths. 
Furthermore, in the \osp, 
the cost of an arm (road) represents the traveling time on the road, and it dynamically changes due to the time-varying amount of traffic or accidents (e.g., cyber attacks in the case of the \osp\ on communication networks). In this paper, we focus on the adversarial \cmab. 
The main difficulty with the adversarial \cmab\ is that the size of the decision set is generally exponential in the number of arms.
To handle huge decision sets, existing methods for this problem assume that the decision set has certain properties. 
One such method is 
\combexp~\cite{combes2015combinatorial}, which 
can cope with the difficulty if the decision set consists of, 
for example, 
sets of arms satisfying a size constraint,  
or matchings on a given network.
However, it has been hard to design practical algorithms for adversarial \cmab\ instances with complex decision sets defined on networks; 
for example, the \osp, \dst, and \cg\ on undirected networks.

In this paper, we develop a practical and theoretically guaranteed algorithm for the adversarial \cmab, 
which is particularly effective for 
network-based adversarial \cmab\ 
instances.
We first propose \combandms\ (\comband~\cite{cesa2012comband} with Weight Modification), which is theoretically 
guaranteed to achieve either $O(T^{2/3})$ {\it regret} with high probability or $O(\sqrt{T})$ {\it expected regret}, where $T$ is the number of rounds. 
The above bounds are {\it any-time guarantees}~\cite{braun2016efficient}, and we can choose which regret value 
\combandms\ actually achieves by setting its hyper parameter at an appropriate value.
We then show that our \combandms\ can be performed on a {\it compressed} decision set;
we assume that a decision set is given as a {\it zero-suppressed decision diagram} (ZDD)~\cite{minato1993zero}, 
which is a compact graph representation of a family of sets. 
The time and space complexities of \combandms\ with a ZDD 
are linear in the size of the ZDD,  
whereas those of the naive \combandms\ is proportional to the size of a decision set.  
It is known that a ZDD tends to be small if 
it represents a set of subnetworks such as $s$-$t$ paths  or Steiner trees~\cite{kawahara2014frontier}.  
Thus our algorithm is effective for network-based adversarial \cmab\ instances including the \osp, \dst, and 
\cg. 
Experimental results on \osp, \dst, and \cg\ instances show that our algorithm is more scalable than naive algorithms that directly deal with decision sets. 
%
To the best of our knowledge, this is the first work to implement algorithms for the adversarial \cmab\ and provide experimental results, thus revealing the practical usefulness of adversarial \cmab\ algorithms.

\section{Related work}
Many studies have considered the adversarial \cmab\ with specific decision sets, e.g., $m$-sets~\cite{kale2010slate} and permutations~\cite{ailon2014permutation}.
In particular, the \osp, which is a \cmab\ problem on a network with an \st path constraint, has been extensively studied~\cite{awerbuch2004adaptive,gyorgy2007online} due to its practical importance. 
Whereas the previous studies have focused on the \osp\ on directed networks, 
our algorithm is also applicable to the \osp\ on undirected networks. 

The adversarial \cmab\ with general decision sets has been also extensively studied in~\cite{audibert2014regret,braun2016efficient,bartlett2008high,cesa2012towards,cesa2012comband,combes2015combinatorial,uchiya2010adversarial}.
One of the best known algorithms for this problem is \comband~\cite{cesa2012comband}, which has been proved to achieve $O(\sqrt{T})$	expected regret.
Recently the algorithm has been also proved to achieve $O(T^{2/3})$ {\it regret} with high probability 	in~\cite{braun2016efficient};\footnote{
The proof 
seems to include some mistakes. 
However, 
their techniques for the proof are still useful, 
and so we prove the $O(T^{2/3})$ high-probability regret bound of our algorithm by partially modifying their proof;
the modified parts are the description of the algorithm and Lemma~\ref{lem:highprob1} in the supplementary materials.
}
more precisely, the regret of \comband\ is bounded by $O(t^{2/3})$ with high probability in any $t$-th round $(t=1,\dots,T)$, which is called an any-time guarantee.
Although \comband\ has the strong theoretical results, its time complexity generally depends on the size of decision sets, which can be prohibitively large in practice.
To avoid such expensive computation, \combexp~\cite{combes2015combinatorial} scales up \comband\ by employing a projection onto the convex hull of the decision set via KL-divergence.
For some decision sets for which the projection can be done efficiently (e.g., $m$-sets or a set of matchings), \combexp\ 
runs faster than \comband, achieving the same theoretical guarantees.
However, it is difficult to perform the projection for other decision sets (e.g., \st paths or Steiner trees); actually it is NP-hard to do the projection in the case of the \osp\ and \dst\ on undirected networks. 

On the other hand, thanks to recent advances in constructing {\it decision diagrams} (DDs),     optimization techniques using DDs have been attracting much attention~\cite{bergman2016decision,coudert1997solving,morrison2016solving}.
Those techniques are advantageous in that DDs can efficiently store all solutions satisfying some complex constraints; for example, constraints that are hard to represent as a set of inequalities.
The ZDD~\cite{minato1993zero}, which we use in our algorithm, is a kind of DD that is known to 
be suitable for storing specific network substructures (e.g., \st paths or Steiner trees). 
Thus our algorithm with ZDDs runs fast in many \cmab\ instances defined on networks, including the \osp, \dst, and \cg.

\section{Adversarial \cmab}\label{section:background} 


\newcommand{\regret}[1]{\mathrm{R}_{{#1}}}
\newcommand{\eregret}[1]{\overline{\mathrm{R}}_{{#1}}}
\newcommand{\Xsub}[1]{{X_{{#1}}}}
\newcommand{\lsub}[1]{{{\text{\boldmath{$\ell$}}}_{{#1}}}}
\newcommand{\csub}[1]{{c_{{#1}}}}


We here define the adversarial \cmab, which is a sequential decision problem consisting of $T$ rounds.
Let $[m] := \{1,\dots,m\}$ for any $m \in \N$.
We use $E = [d]$ to denote a set of arms and also use $\Fl \subseteq 2^E$ to denote a {decision set}, where $X \in \Fl$ is a {super arm}.
At each $t$-th round ($t\in[T]$), 
an {adversary} secretly defines a {\it loss vector} $\lb_t:=(\losst{1},\dots,\losst{i})^\top\in\R^d$
and
a {player} chooses a super arm $X_t\in\Fl$.
Then, the player incurs and observes the cost $c_t=\lb^\top_t\ib_\Xt$, 
where $\ib_\Xt \in \{0, 1\}^d$ is an indicator vector such that its $i$-th element is $1$ if $i \in \Xt$ and $0$ otherwise.
Note that the player {\it cannot} observe $\lb_t$.
The aim of the player is to minimize the {regret} $\regret{T}$ defined as follows:
\begin{align*}
\regret{T}:=
\sum_{t=1}^T
\lb_t^\top\ib_\Xt - \min_{X\in\Fl}\sum_{t=1}^T\lb_t^\top\ib_X.
\end{align*}
The first term is the cumulative cost and the second term is the total cost of the best single super arm selected with hindsight. 
Namely, $\regret{T}$ expresses the extra cost that the player incurs against the best single super arm.\footnote{
If the adversary behaves adaptively, 
the above interpretation of the regret is somewhat inappropriate; 
in such cases using the {\it policy regret}~\cite{arora2012online} 
is considered to be more suitable. 
However, we here focus on  
the above regret and expected regret, 
leaving an analysis based on the policy regret for future work.}
As is customary,
we assume $\max_{t\in[T], X\in\Fl}|\lb^\top_t\ib_X|\le 1$.

If the adversary and/or the player choose $\lsub{t}$ and $\Xsub{t}$ in a stochastic manner, then $\regret{t}$ is a random variable of a joint distribution $p(\lsub{1:t}, \Xsub{1:t})$,
where $\Xsub{1:t} = \{X_1,\dots,X_t\}$ and $\lsub{1:t} = \{\lb_1,\dots,\lb_t\}$.
In the adversarial \cmab, $p$ is assumed to satisfy the following conditional independence: 
$p(\lsub{1:T},\Xsub{1:T}) = \prod_{t \in [T]} p(\Xsub{t} \mid \lsub{1:t-1}, \Xsub{1:t-1})p(\lsub{t} \mid \lsub{1:t-1},\Xsub{1:t-1})$, where $\Xsub{1:0} = \lsub{1:0} = \{\}$. 
$p(\lsub{t} \mid \Xsub{1:t-1}, \lsub{1:t-1})$ corresponds to the adversary's strategy and $p(\Xsub{t} \mid \Xsub{1:t-1}, \lsub{1:t-1})$ corresponds to the player's strategy.
Since the player cannot directly observe $\lsub{1:t}$, the player's strategy must satisfy $p(\Xsub{t} \mid \Xsub{1:t-1}, \lsub{1:t-1}) = p(\Xsub{t} \mid \Xsub{1:t-1}, \csub{1:t-1})$.
Using the joint distribution $p$, the expected regret $\eregret{T}$ is defined as follows:
\[
\eregret{T}:=\E_{\lsub{1:T},\Xsub{1:T}\sim p} [\regret{T}]. 
\]
The objective of the adversarial \cmab\ is to design the player's strategy
$p(\Xsub{t} \mid \Xsub{1:t-1}, \csub{1:t-1})$ so that it minimizes $\regret{T}$ or $\eregret{T}$.
In this paper, we use $p_t(X_t)$ as shorthand for $p(\Xsub{t} \mid \Xsub{1:t-1}, \csub{1:t-1})$.


\section{Proposed algorithm for adversarial \cmab}
\label{section:algorithm}

We here propose \combandms\ (\comband\ with Weight  Modification), which is an algorithm for designing the player's strategy $p_t(\Xsub{t})$ with strong theoretical guarantees as described later. 
Algorithm~\ref{alg:comband} gives the details of \combandms. 
In what follows, 
we define $\bdx:=\max_{X\in\Fl} \norm{\ib_X}$ for any given $\Fl\subseteq 2^E$, where $\norm{\cdot}$ is the Euclidian norm.  
We also define $\lambda$ as the smallest non-zero eigenvalue of $\E_{X\sim u}[\ib_X\ib_X^\top]$, 
where $u$ is the uniform distribution over $\Fl$.  

	\begin{algorithm}
	\caption{\combandms$(\alpha, \Fl)$}\label{alg:comband}
	\begin{algorithmic}[1]
		\STATE $\tlweight{1}{i}\gets1$ ($i\in E$)
		\FOR{$t=1,\dots,T$}
		\STATE $\gamma_t\gets \frac{t^{-1/\a}}{2}$, $\eta_t\gets\frac{\lambda t^{-1/\a}}{2\bdx^2}$, $\eta_{t+1}\gets\frac{\lambda (t+1)^{-1/\a}}{2\bdx^2}$
		\STATE $X_t \sim p_t$
		\STATE $c_t\gets\lb_t^\top\ib_{\Xt}$ ($\lb_t$ is unobservable)
		\STATE $P_t(i,j)\gets \sum_{X\in\Fl : i,j\in X} p_t(X)$ ($i,j\in[d])$
		\STATE $\lbh_t\gets c_tP_t^+\ib_{\Xt}$
		\STATE $\tlweight{t+1}{i}\gets \tlweight{t}{i}^{\eta_{t+1}/\eta_{t}}
		\exp\big(-\eta_{t+1}\hloss{t}{i}\big)$ ($i\in E$)
		\ENDFOR
		\RETURN $\{X_t \mid t \in [T]\}$
	\end{algorithmic}
	\end{algorithm}


Given an arbitrary non-negative vector $\wb=(w_1,\dots,w_d)^\top \in \R^d$ and a decision set $\Fl \subseteq 2^{E}$, we define the {\it constrained distribution} $p(X ; \wb, \Fl)$ over $\Fl$ as follows:
\begin{align}
p(X ; \wb, \Fl) &:= \frac{w(X)}{Z(\wb, \Fl)},
&
Z(\wb, \Fl) &:= \sum_{X \in \Fl} w(X),
&
w(X) &:= \prod_{i \in X} w_i.
\label{eq:constrained_distribution}
\end{align}
Using the above,
we define the player's strategy $p_t(\Xsub{t})$, 
which appears in Step~4, as follows:
\begin{align}
p_t(\Xsub{t}):=(1-\gamma_t)p(\Xsub{t} ; \wbt_t, \Fl)+\gamma_t p(\Xsub{t} ; \ib_E, \Fl), \label{def:dists}
\end{align}
where $\wbt_t=(\tlweight{t}{1},\dots,\tlweight{t}{d})^\top$ is the weight vector defined in Step~8, and $\gamma_t$ is the parameter defined in Step~3; we  note that $p(\Xsub{t};\ib_E,\Fl)$ is the uniform distribution over $\Fl$. 
Thus $p_t$ is a mixture of two constrained distributions with the mixture rate $\gamma_t$.

Given a distribution $p$ over $\Fl$, 
a matrix $P$ is called a 
{\it co-occurrence probability matrix} (CPM) 
if its $(i,j)$ entry $P(i,j)$ is given by the 
{\it co-occurence probability} $p(i\in X,j\in X):=\sum_{X\in\Fl:i,j\in X} p(X)$. 
The matrix $P_t$ computed in Step~6 is the CPM of $p_t$, 
and $P_t^+$ used in Step~7 is the pseudo-inverse of $P_t$.
From Eq.~\eqref{def:dists}, the following equation holds:
\begin{equation}\label{eq:cmppt}
P_t(i,j)=(1-\gamma_t)p(i\in X, j\in X ; \wbt_t, \Fl)+\gamma_t p(i\in X, j\in X ; \ib_E, \Fl).
\end{equation}
%

The above \combandms\ is based on 
\comband~\cite{cesa2012comband};
if we replace Step~8 of \combandms\ with $\tlweight{t+1}{i}\gets\tlweight{t}{i}\exp(-\eta_t\hloss{t}{i})$, \combandms\ corresponds perfectly to the original \comband. 
Hence the one and only one difference is the weight modification in Step~8. 
However, introducing this weight modification gives us the following theoretical guarantees 
(for proofs, see the supplementary materials): 
%
\begin{thm}\label{thm:highprob}
For any $\Fl$, \combandms$(\alpha=3, \Fl)$ achieves $\regret{T}\le O\Big(\Big(\frac{d\lambda}{\bdx^2} + \sqrt{\frac{\bdx^2}{ \lambda} \ln\frac{|\Fl|+2}{\delta} }\Big)T^{2/3}\Big)$
with probability at least $1-\delta$.
\end{thm}
\begin{thm}\label{thm:expectation}
For any $\Fl$, \combandms$(\alpha = 2, \Fl)$ achieves
$\eregret{T}\le
O\Big(
\Big(
\frac{d\lambda}{\bdx^2}
+
\frac{\bdx^2\ln |\Fl|}{\lambda}
\Big)
\sqrt{T}\Big)$.
\end{thm}
In other words, \combandms\ achieves either $O(T^{2/3})$ regret with high probability or
$O(\sqrt{T})$ expected regret 
as an any-time guarantee 
by choosing the hyper parameter $\a$ appropriately. 

There are two difficulties when it comes to performing \combandms; 
the first is sampling from the player's strategy $p_t(X_t)$ (Step~4), and the second is computing the CPM $P_t$ (Step~6).
Naive methods for sampling from $p_t$ and computing $P_t$ require $O(|\Fl|)$ and $O(d^2|\Fl|)$ computation times, respectively, where $|\Fl|$ is generally exponential in $d$, and so are the time complexities. 
In the following section, we propose efficient methods for sampling from any given constrained distribution $p(X ; \wb, \Fl)$ and for computing the CPM of $p(X ; \wb, \Fl)$.
Because $p_t$ is a mixture of two constrained distributions, we can efficiently sample from $p_t$ and compute the CPM of $p_t$ using the proposed methods. 
\section{\combandms\ on compressed decision sets}
\label{section:implementation}

As shown above, 
\combandms\ requires sampling from constrained distributions and computing CPMs as its building blocks, which generally require $O(|\Fl|)$ and $O(d^2|\Fl|)$ computation times, respectively. 
Moreover, computing $\bdx$ can also require $O(|\Fl|)$ time.
Those computation costs can be prohibitively expensive since $|\Fl|$ is generally exponential in $d$. 
In this section, we present efficient algorithms for the building blocks that are based on dynamic programming (DP) on a ZDD, which is a compressed representation of $\Fl$.
We first briefly describe ZDDs and then propose two DP methods 
for sampling and computing CPMs.
$\bdx$ can also be computed in a DP manner on a ZDD.

\subsection{Zero-suppressed binary decision diagrams (ZDDs)}
\label{subsection:ZDD} 

A ZDD \cite{minato1993zero} is a compact graph representation of a family of sets.
Given $\Fl \subseteq 2^{E}$, a ZDD for $\Fl$ is a directed acyclic graph (DAG) denoted by $\zdd{\Fl} = (V, A)$, where $V = \{0,1,\dots,|V|-1\}$ is a set of vertices and $A \subseteq V \times V$ is a set of directed arcs.
%
$\zdd{\Fl}$ contains one root vertex $r\in V$ and two terminal vertices: 1-terminal and 0-terminal.
Without loss of generality, we assume that $V=\{0,\dots,|V|-1\}$ is arranged in a topological order of $\zdd{\Fl}$; 
$r=|V|-1$ holds and  
the $b$-terminal ($b \in \{0,1\}$) is 
denoted simply by $b\in V$. 
%
Each non-terminal vertex $v \in V \backslash \{0,1\}$ is labeled by an integer in $E$ and has exactly two outgoing arcs: 1-arc and 0-arc.
A vertex pointed by the $b$-arc of $v$ is called the $b$-child of $v$.
We use $\lbl{v}$, $\arc{v}{b}$, and $\child{v}{b}$ to denote $v$'s label, $b$-arc, and $b$-child, respectively.
Consequently, $\arc{v}{b} = (v, \child{v}{b})$ holds.
%
We use $\Rt{v}{u}$ ($v,u \in V$) to denote a set of routes (directed paths) from $v$ to $u$ on $\zdd{\Fl}$, where a route $R \in \mathcal{R}_{v,u}$ is a set of directed arcs: $R \subseteq A$.
Given $R \in \Rt{v}{u}$, we define $X(R) \subseteq E$ as $X(R) := \{\lbl{ {v^{\prime}} } \mid (v^{\prime}, \child{ v^{\prime} }{1}) \in R \}$.
Then, $\zdd{\Fl}$ satisfies 
\begin{align}
\Fl = \{X(R) \mid R \in \mathcal{R}_{r,1} \}.
\label{eq:zdd}
\end{align}
Therefore, $\zdd{\Fl}$ represents the decision set $\Fl$ as a set of all routes from its root $r$ to the 1-terminal.
Note that once $\zdd{\Fl}$ is obtained, $\bdx=\max_{X\in\Fl} \sqrt{|X|}$ is easily computed by 
a DP method to find $R \in \mathcal{R}_{r,1}$ that maximizes $|X(R)|$.

In general, a ZDD is assumed to be {\it ordered} and {\it reduced}.
$\zdd{\Fl}$ is said to be {\it ordered} if $v > u \Rightarrow \lbl{v} < \lbl{u}$ holds for all $v, u \in V \backslash \{0, 1\}$.
%
A non-terminal vertex $v$ is said to be {\it redundant} if $\child{v}{1} = 0$: 
its 1-arc directly points to the 0-terminal.
A redundant vertex $v$ can be removed by 
replacing all $(u, v) \in A$ with $(u, \child{v}{0})$
without loss of the property~\eqref{eq:zdd}.
%
A non-terminal vertex $v$ is said to be {\it sharable} if there exists another vertex $v^\prime$ such that $\lbl{v} = \lbl{v^\prime}$ and $\child{v}{b} = \child{v^\prime}{b}$ ($b \in \{0,1\})$:
$v$ and $v^\prime$ have the same label and children.
A sharable vertex $v$ can be removed by replacing $(u, v) \in A$ with $(u, v^\prime)$.
%
$\zdd{\Fl}$ is said to be {\it reduced} if no vertex is redundant or sharable.
In this paper, we assume that $\zdd{\Fl}$ is ordered and reduced.
We show an example of a ZDD in Figure~\ref{fig:zddfigs}.

%
\iffigure  
\begin{wrapfigure}[22]{r}{0.4\linewidth} 
	\begin{minipage}[b]{0.49\linewidth}
	\begin{center}
		\includegraphics[width=0.9\linewidth, bb=0 0 540 720]{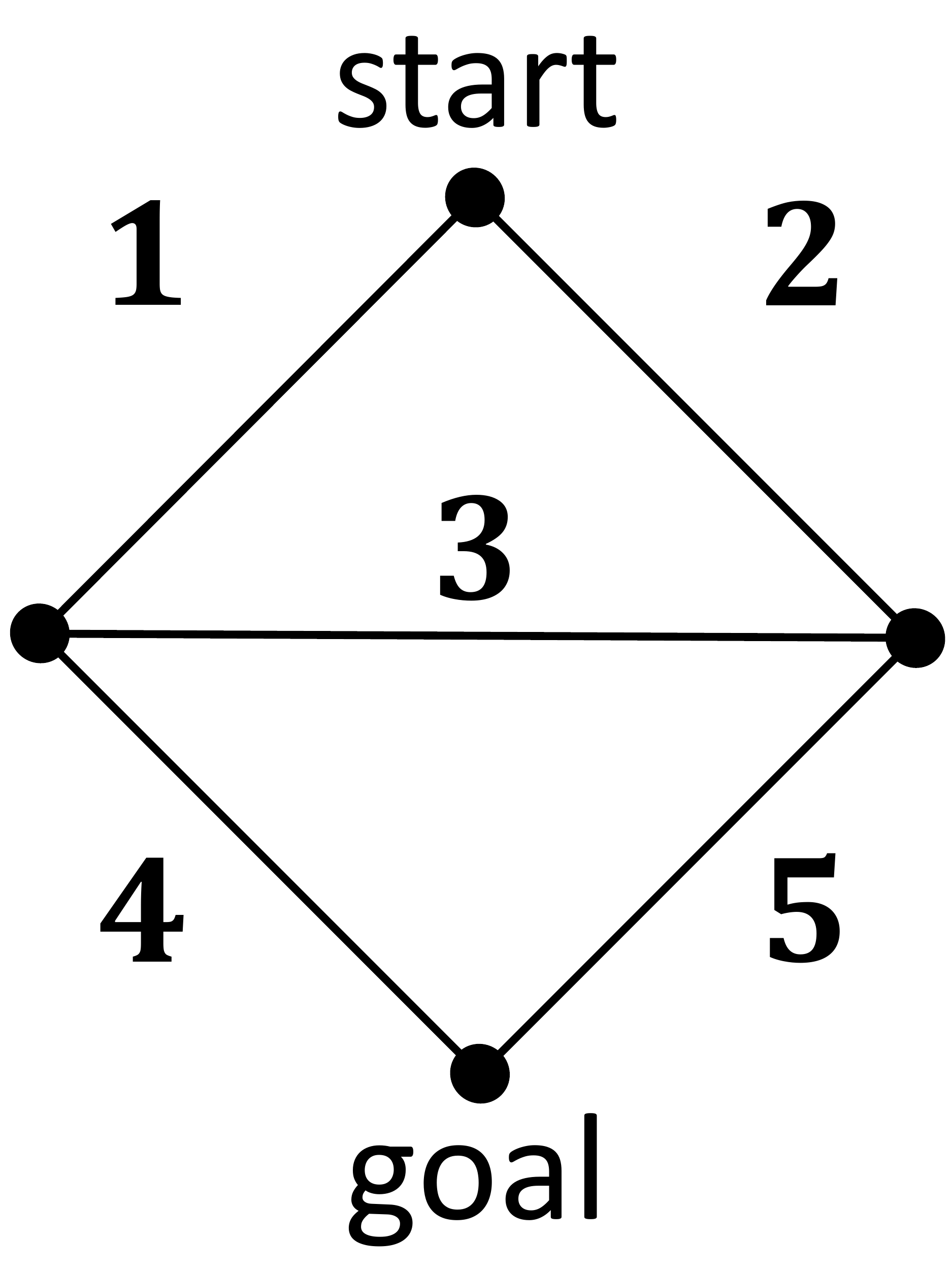}
		\subcaption{}
	\end{center}
	\end{minipage}
	\begin{minipage}[b]{0.49\linewidth}
	\begin{center}
		\includegraphics[width=0.9\linewidth, bb=0 0 540 720]{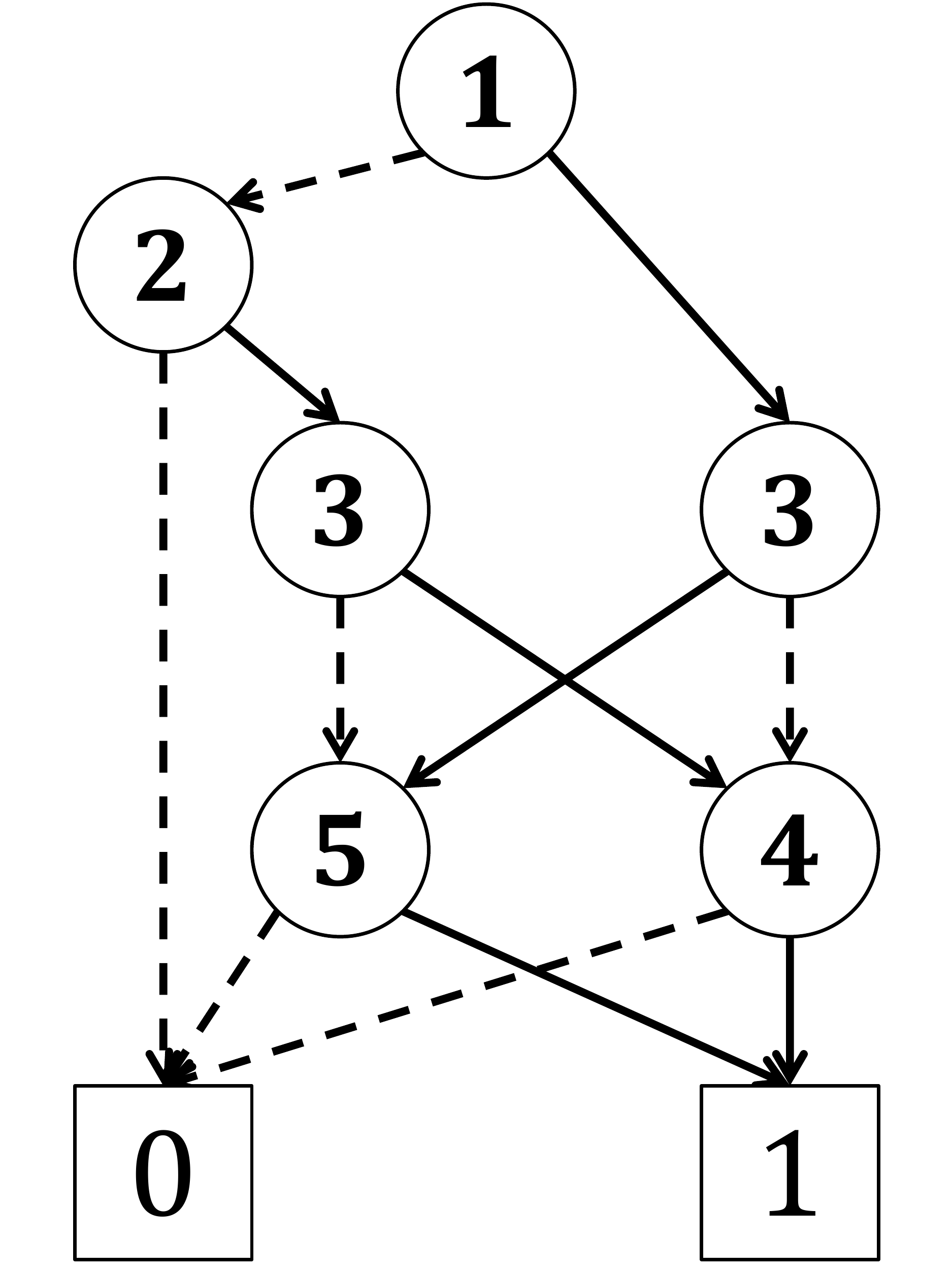}
		\subcaption{}
	\end{center}
\end{minipage}
\caption{
(a) An example network with an edge set $E=\{1,\dots,5\}$,
and (b) a ZDD that stores all paths from the start to the goal;
each non-terminal vertex $v$ is labeled $\lbl{v} \in E$, 
and 
\zarc s and \oarc s are indicated by dashed and solid lines, 
respectively. 
Note that we have $\Fl = \{X(R) \mid R \in \mathcal{R}_{r,1} \}=\{ \{1,4\}, \{2,5\}, \{1,3,5\}, \{2,3,4\}  \}$. }
\label{fig:zddfigs}
\end{wrapfigure}
\fi
%
The ZDDs are known to store various families of sets 
compactly in many applications. 
In particular, if a decision set is a set of specific network substructures 
(e.g, a set of \st paths or Steiner trees), 
the ZDD representing the decision set 
tends to be small. 
As we will see later, 
the time complexity of \combandms\ with a ZDD $\zdd{\Fl}=(V,A)$ 
is $O(d|V|)$, 
and so it runs fast if the ZDD 
is small. 
In theory, 
if $\Fl$ is a set of specific network substructures, 
then $|V|$ is bounded by a value that is exponential 
in the {\it pathwidth}~\cite{inoue2016acceleration}. 
Thus, even if a network-based decision set 
is exponentially large in $d$,  
the time complexity of our algorithm in each round can be polynomial in $d$ if the pathwidth of the network 
is bounded by a small constant. 
%
%

The {\it frontier-based search}~\cite{kawahara2014frontier}, 
which is based on Knuth's {\it Simpath} algorithm~\cite{knuth2011art1}, 
has recently received much attention as a fast top-down construction algorithm for ZDDs that represent a family of 
subnetworks.
In practice ZDDs are easily obtained via existing software~\cite{inoue2016graphillion} for various network-based constraints.
In this paper, we omit the details of ZDD construction and assume that a 
decision set $\Fl$ 
is represented by a ZDD $\zdd{\Fl}$ rather than by the explicit enumeration of the components of $\Fl$.

\subsection{Sampling from constrained distributions}
%

We here propose an efficient algorithm for sampling from a constrained distribution $p(X ; \wb, \Fl)$.
We first introduce the following {\it forward weight} (FW) $\F{v}$ and {\it backward weight} (BW) $\B{v}$ ($v \in V$):
\begin{align}
\F{v} &:= \sum_{R \in \Rt{r}{v}} w(R),
&
\B{v} &:= \sum_{R \in \Rt{v}{1}} w(R),
\label{eq:backward_forward_weight}
\end{align}
where $w(R)$ is an abbreviation of $w(X(R))=\prod_{i\in X(R)}w_i$.
By combining Eq.~\eqref{eq:constrained_distribution}, \eqref{eq:zdd}, and \eqref{eq:backward_forward_weight}, we obtain $Z(\wb, \Fl) = \B{r} = \F{1}$.
$\B{} := \{\B{0},\dots,\B{r}\}$ and $\F{} := \{\F{0},\dots,\F{r}\}$ can be efficiently computed in a dynamic programming manner on $\zdd{\Fl}$ as shown in Algorithm FW($\zdd{\Fl}, \wb$) and BW($\zdd{\Fl}, \wb$).
Once we obtain $\B{}$, we can draw a sample from $p(X ; \wb, \Fl)$ by top-down sampling on $\zdd{\Fl}$ 
without rejections 
as shown in Algorithm Draw($\zdd{\Fl}, \wb, \B{}$), 
where $\Ber{\theta}$ is the Bernoulli distribution 
with the parameter $\theta\in[0,1]$. 
The space and time complexity when computing $\F{}$ and $\B{}$ is proportional to $|V|$.
%
This constrained sampling is based on the same idea as that used in logic-based probabilistic modeling \cite{Ishihata08,Ishihata11}.

\begin{table}[htb]
	\centering
	\begin{tabular}[t]{ccc}
		\begin{minipage}[t]{0.3\textwidth}
			\label{alg:FW}
			\begin{algorithmic}[1]
				\STATE {\bf Algorithm} FW($\zdd{\Fl}, \wb$)
				\STATE $\F{r} \leftarrow 1$
				\STATE $\F{v} \leftarrow 0$ ($\forall v \in V \backslash \{r\}$)
				\FOR{$v = r,\dots,2$}
				\STATE $\F{ \child{v}{0} } +\!= \F{v}$
				\STATE $\F{ \child{v}{1} } +\!= w_{\lbl{v}} \F{v}$
				\ENDFOR
				\STATE $\F{} := \{\F{0},\dots,\F{r}\}$
				\RETURN $\F{}$
			\end{algorithmic}
		\end{minipage}
		&
		\begin{minipage}[t]{0.3\textwidth}
			\begin{algorithmic}[1]
				\label{alg:BW}
				\STATE {\bf Algorithm} BW($\zdd{\Fl}, \wb$)
				\STATE $\B{1} \leftarrow 1$
				\STATE $\B{v} \leftarrow 0$ ($\forall v \in V\backslash\{1\}$)
				\FOR{$v = 2,\dots,r$}
				\STATE $\B{v} \leftarrow \B{ \child{v}{0} } + w_{\lbl{v}} \B{ \child{v}{1} }$
				\ENDFOR
				\STATE $\B{} := \{\B{0},\dots,\B{r}\}$
				\RETURN $\B{}$
			\end{algorithmic}
		\end{minipage}
		&
		\begin{minipage}[t]{0.35\textwidth}
			\begin{algorithmic}[1]
				\STATE {\bf Algorithm} Draw($\zdd{\Fl}, \wb, \B{}$)
				\STATE $X \leftarrow \{\}$, $v \leftarrow r$
				\WHILE{$v > 1$}
				\STATE $\theta \leftarrow w_{\lbl{v}} \B{ \child{v}{1} } / \B{v}$
				\STATE $b \sim \Ber{\theta}$
				\STATE $X \leftarrow X \cup \{ \lbl{v} \} $ ~if~ $b = 1$
				\STATE $v \leftarrow \child{v}{b}$
				\ENDWHILE
				\RETURN $X$
			\end{algorithmic}
		\end{minipage}
	\end{tabular}
\end{table}

\subsection{Computing co-occurrence probabilities}

Given a constrained distribution $p(X ; \wb, \Fl)$, 
we define $\CP{i}{j} := p(i\in X, j \in X ; \wb, \Fl)$ as 
the co-occurrence probability of $i$ and $j$ ($i,j \in E$). 
We here propose an efficient algorithm for computing $\CP{i}{j}$ ($i\le j$), 
which suffices for obtaining $\CP{i}{j}$ for all $i,j\in[d]$ since $\CP{i}{j}=\CP{j}{i}$.
Using Eq.~\eqref{eq:constrained_distribution} 
and the notion of $\zdd{\Fl}$, $\CP{i}{j}$ can be written as follows:
\begin{align}
\CP{i}{j}
&= 
\sum_{R \in \Rt{r}{1} : i,j \in X(R)}
\frac{w(R)}{Z(\wb, \Fl)}.
\label{eq:CP}
\end{align}
%

We first consider $\CP{i}{i}$ as a special case of $\CP{i}{j}$.
By combining Eq.~\eqref{eq:backward_forward_weight} and \eqref{eq:CP}, we obtain
\begin{align*}
\CP{i}{i}
&=
\sum_{R \in \Rt{r}{1} : i \in X(R)}
\frac{w(R)}{Z(\wb, \Fl)}
=
\sum_{v \in V : \lbl{v}=i}
\sum_{
\substack{
R^\prime \in \Rt{r}{v} \\ 
R^{\prime\prime} \in \Rt{\child{v}{1}}{1}
}
}
\frac{ w( R^\prime \cup \{i\} \cup R^{\prime\prime} ) }{\B{r}}
=
\sum_{v \in V : \lbl{v}=i} 
\frac{\F{v} w_i \B{ \child{v}{1} }}{\B{r}}.
\end{align*}
%
Next, to compute $\CP{i}{j}$ $(i < j)$, 
we rewrite the right hand side of Eq.~\eqref{eq:CP} 
using the {\it backward weighted co-occurrence} (BWC) $\BC{v}{j}$ ($j \geq \lbl{v}$) as follows:
\begin{align*}
\CP{i}{j}
=
\sum_{v \in V : \lbl{v}=i}
\frac{\F{v} w_i \BC{ \child{v}{1} }{j}}{ \B{r} },
&&
\BC{v}{j}
&:=
\sum_{R \in \Rt{v}{1} : j \in X(R)}
w(R).
\end{align*}
Because $\BC{v}{j}$ is a variant of $\B{v}$, $\mathrm{C} := \{ \BC{v}{j} \mid v \in V, j \geq \lbl{v} \}$ can be computed in a similar manner to $\B{}$ as shown in Algorithm BWC($\zdd{\Fl}, \wb, \B{}$).
To conclude,  $\CP{i}{j}$ can be computed by Algorithm CPM($\zdd{\Fl}, \wb, \F{}, \B{}, \mathrm{C}$).
The total space and time complexity of computing $\mathrm{P} := \{\CP{i}{j} \mid i \leq j\}$ is $O(d|V|)$.

\begin{table}[htb]
\centering
\begin{tabular}[t]{cc}
\begin{minipage}[t]{0.45\textwidth}
\begin{algorithmic}[1]
\STATE {\bf Algorithm} BWC($\zdd{\Fl}, \wb, \B{}$)
\STATE $\BC{v}{j} \leftarrow 0$ ($\forall v \in V$, $\forall j \in E$)
\FOR{$v = 2,\dots,r$}
\STATE $\BC{v}{\lbl{v}} \leftarrow w_{\lbl{v}} \B{ \child{v}{1} }$ 
\FOR{$j = \lbl{v}+1,\dots,d$}
\STATE $\BC{v}{j} \leftarrow \BC{ \child{v}{0} }{j} + w_{\lbl{v}} \BC{ \child{v}{1} }{j}$
\ENDFOR
\ENDFOR
\STATE $\mathrm{C} := \{\BC{v}{j} \mid v \in V, j \geq \lbl{v}\}$
\RETURN $\mathrm{C}$
\end{algorithmic}
\end{minipage}
&
\begin{minipage}[t]{0.45\textwidth}
\begin{algorithmic}[1]
\STATE {\bf Algorithm} CPM($\zdd{\Fl}, \wb, \B{}, \F{}, \mathrm{C}$)
\STATE $\CP{i}{j} \leftarrow 0$ ($\forall i,j \in E$)
\FOR{$v = 2,\dots,r$}
\STATE $i \leftarrow \lbl{v}$
\STATE $\CP{i}{i} +\!= \F{v} w_i \B{ \child{v}{1} } / \B{r}$
\FOR{$j = i+1,\dots,d$}
\STATE $\CP{i}{j} +\!= \F{v} w_i \BC{ \child{v}{1} }{j} / \B{r}$
\ENDFOR 
\ENDFOR
\STATE $\mathrm{P} := \{\CP{i}{j} \mid i,j \in [d], i \leq j\}$
\RETURN $\mathrm{P}$
\end{algorithmic}
\end{minipage}
\end{tabular}
\end{table}

\section{Experiments}

We applied our \combandms\ with ZDDs to three network-based \cmab\ problems: the \osp, \dst, and \cg.
In the \osp\ and \dst, we used artificial networks to observe the scalability of our algorithm.
In the CG, we used two real-world networks to show the practical utility of our algorithm.
We implemented our algorithm in the C programming language and used Graphillion~\cite{inoue2016graphillion} to obtain the ZDDs. 
We note that constructing ZDDs with the 
software is not a drawback;  
in all of the following instances   
a ZDD was obtained within at most several seconds. 

\subsection{\osp\ and \dst\ on artificial networks}
\newcommand{\unif}{u(0,1)}
\newcommand{\mb}{\mbox{\boldmath $\mu$}}

{\bf Experimental Setting}:
We applied our \combandms\ with ZDDs to the \osp\ and \dst\ instances on artificial networks, 
which are undirected grid networks with $3\times m$ nodes ($m = 3,\dots,10$). 
In both problems, an arm corresponds to an edge of the given network.
In the \osp, a decision set $\Fl$ is a set of all $s$-$t$ paths from the starting node $s$ to the goal node $t$ that are placed on diagonal corners of the given grid.
In the \dst, $\Fl$ is a set of all Steiner trees that contains the four corners of the grid.
The aim of the player is to minimize the cumulative cost of the selected subnetworks over some time horizon.
%
In this experiment, we define the loss vector $\lb_t$ as follows:
We first uniformly sample $\mb_0$ from $[0,1]^d$.
In the $t$-th round, we set $\mb_t=\mb_{t-1}$ with probability $0.9$ or draw a new $\mb_t$ uniformly from $[0,1]^d$ with probability $0.1$.
Then, for each $i \in E$, we draw $h_i\sim\Ber{\mu_{t,i}}$ and set $\loss{t}{i} = 1/d$ if $h_t = 1$ otherwise $-1/d$.
This setting is a stochastic \cmab\ with distributions $\Ber{\mu_{t,i}}$ in the short run, but the adversary secretly reset $\mb_{t}$ with probability $0.1$ in each round to foil the player.



{\bf Compression Power}:
We first assess the compression power of ZDDs
constructed for the decision sets of the \osp\ and \dst\ instances. 
Table~\ref{table:zddsize} shows the sizes of decision sets $\Fl$ and those of the corresponding ZDDs $\zdd{\Fl}$.
In both problems, the ZDD size, $|V|$, grows much more slowly than $|\Fl|$.
In particular, with the \dst\ on the $3 \times 10$ grid, we see that $|V|$ is five orders of magnitude smaller than $|\Fl|$. 
In such cases, our \combandms, which only deals with a ZDD $\zdd{\Fl}$, 
is much more scalable than the naive method that directly deals with $\Fl$.

{\bf Empirical Regret}:
We next show that the empirical regrets of our \combandms\ and \comband\ actually grow sublinearly, 
where \comband\ is also performed on ZDDs. 
We applied these algorithms to the \osp\ and \dst\ on a $3 \times 10$ grid and computed their empirical regrets over a time horizon.
Figures~\ref{fig:experiments}~(a) and~(b) summarize their regrets for the \osp\ and \dst, respectively.
We see that all of the algorithms achieved more or less the same sublinear regrets. 
It was confirmed that all of the regret values were lower than those of the theoretical bounds stated in Theorem~\ref{thm:highprob} and Theorem~\ref{thm:expectation}; 
the precise values of the bounds are provided in the supplementary materials. 

\begin{table}[htb]
\caption{
The sizes of decision sets $\Fl$ and the corresponding ZDDs $\zdd{\Fl}$ for the \osp\ and \dst\ 
(numbers with more than six digits are rounded to three significant digits).
}
\label{table:zddsize}
\vspace{-10pt}
\begin{small}
\begin{center}
\begin{tabular}{ccrrrrrrrr}
& $m$ & 3 & 4 & 5 & 6 & 7 & 8 & 9 & 10 \\
\hline
\multirow{2}{*}{\osp}
& $|\Fl|$
& 12
& 38
& 125
& 414
& 1,369
& 4,522
& 14,934
& 49,322
\\
& $|V|$
& 31
& 76
& 183
& 451
& 1,039
& 2,287
& 4,991
& 11,071
\\ 
\hline
\multirow{2}{*}{\dst}
& $|\Fl|$
& 266
& 4,285
& 69,814
& 1.14$\times 10^6$ 
& 1.86$\times 10^7$ 
& 3.04$\times 10^8$ 
& 4.97$\times 10^9$ 
& 8.12$\times 10^{10}$ 
\\ 
& $|V|$
& 80
& 304
& 1,147
& 4,616
& 18,032
& 67,484
& 238,364
& 933,394
\end{tabular}
\end{center}
\end{small}
\end{table}

\iffigure  
\begin{figure}[htb] 
	\begin{minipage}[b]{0.35\hsize}
    		\begin{center}
			\includegraphics[width=1.00\textwidth, bb=0 0 576 432]{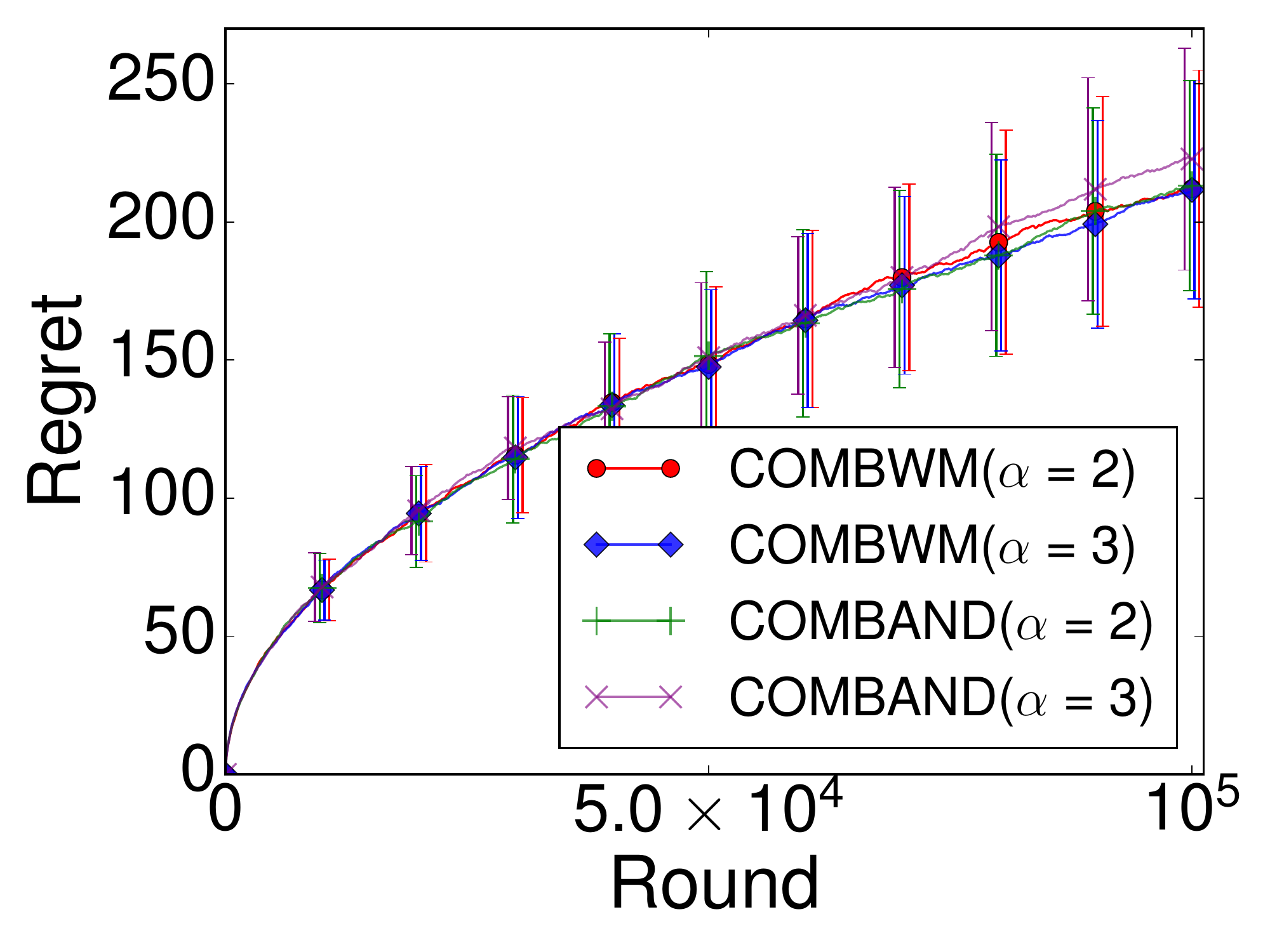}
			\subcaption{\osp\ on the $3\times 10$ grid}
		\end{center}
        		\begin{center}
			\includegraphics[width=1.00\textwidth, bb=0 0 576 432]{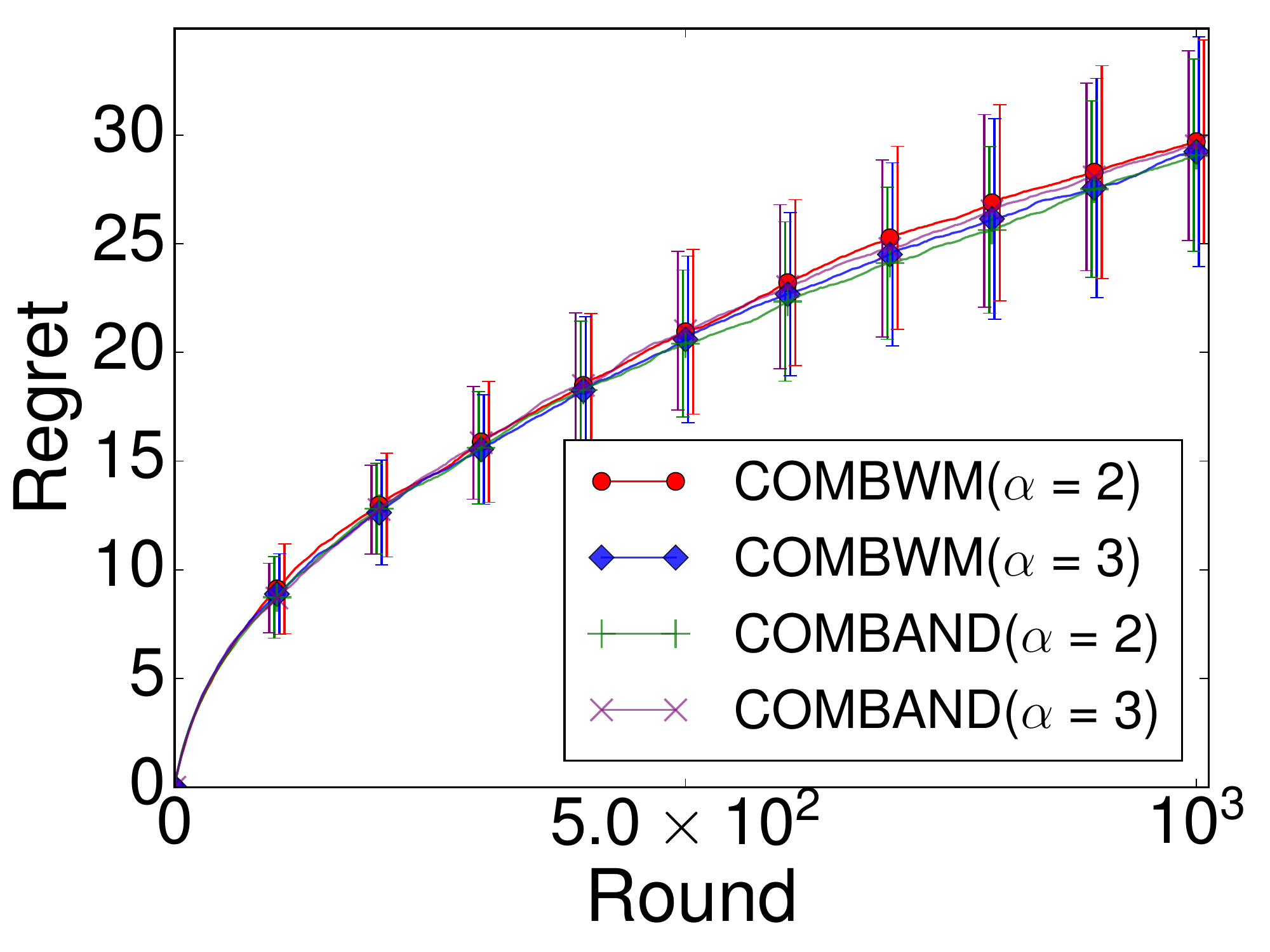}
			\subcaption{\dst\ on the $3\times 10$ grid}
		\end{center}
       \end{minipage}
   	\begin{minipage}[b]{0.35\hsize}
		\begin{center}
			\includegraphics[width=1.0\textwidth, bb=0 0 576 432]{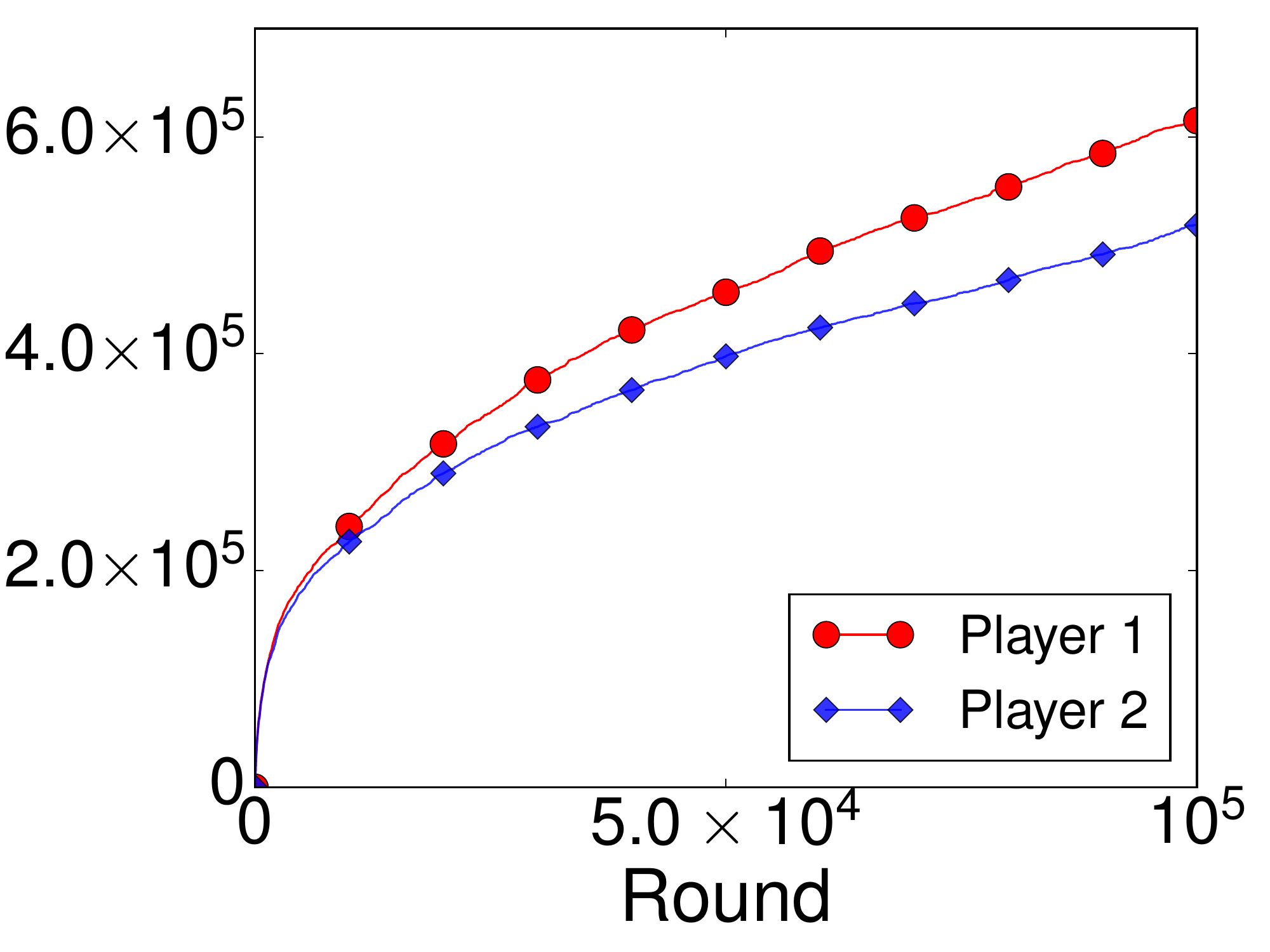}
			\subcaption{CG on MCI}
		\end{center}
       \begin{center}
			\includegraphics[width=1.0\textwidth, bb=0 0 576 432]{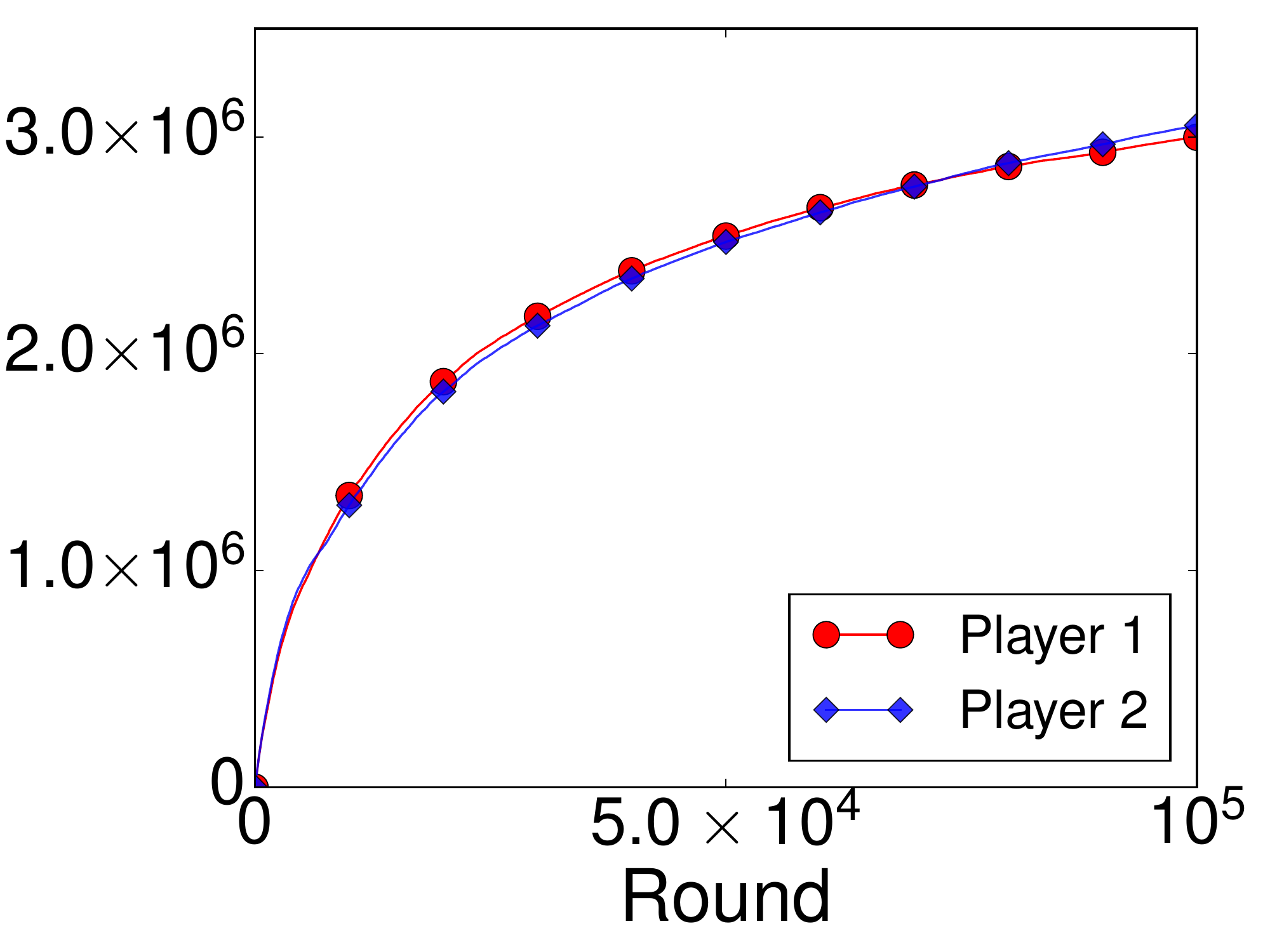}
			\subcaption{CG on ATT}
		\end{center}
	\end{minipage}
	\begin{minipage}[b]{0.28\hsize}
		\begin{center}
		\includegraphics[width=\textwidth, bb=0 0 300 377]{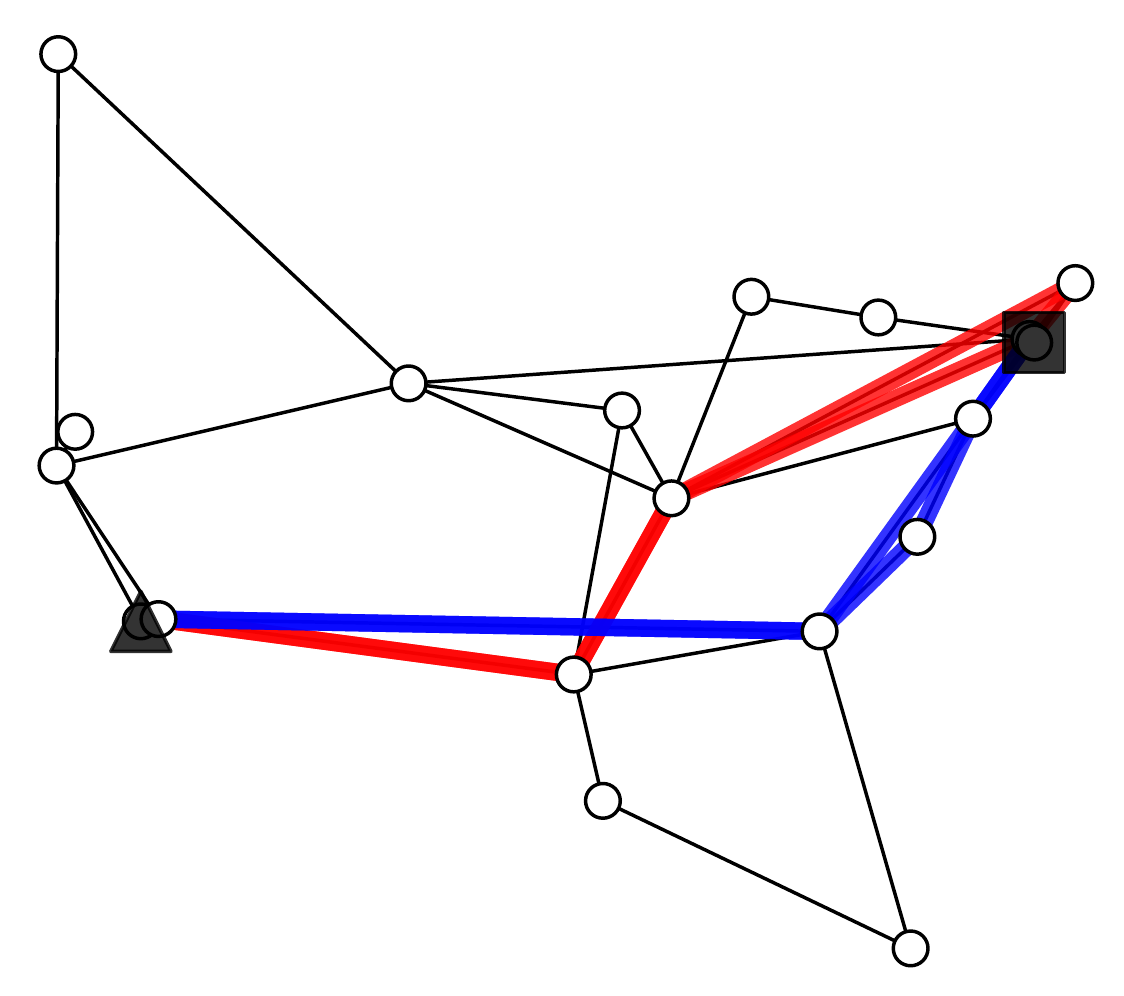} 
			\subcaption{MCI}
		\end{center}
		\vspace{-52pt}
		\begin{center}
		\includegraphics[width=\textwidth, bb=0 0 300 379]{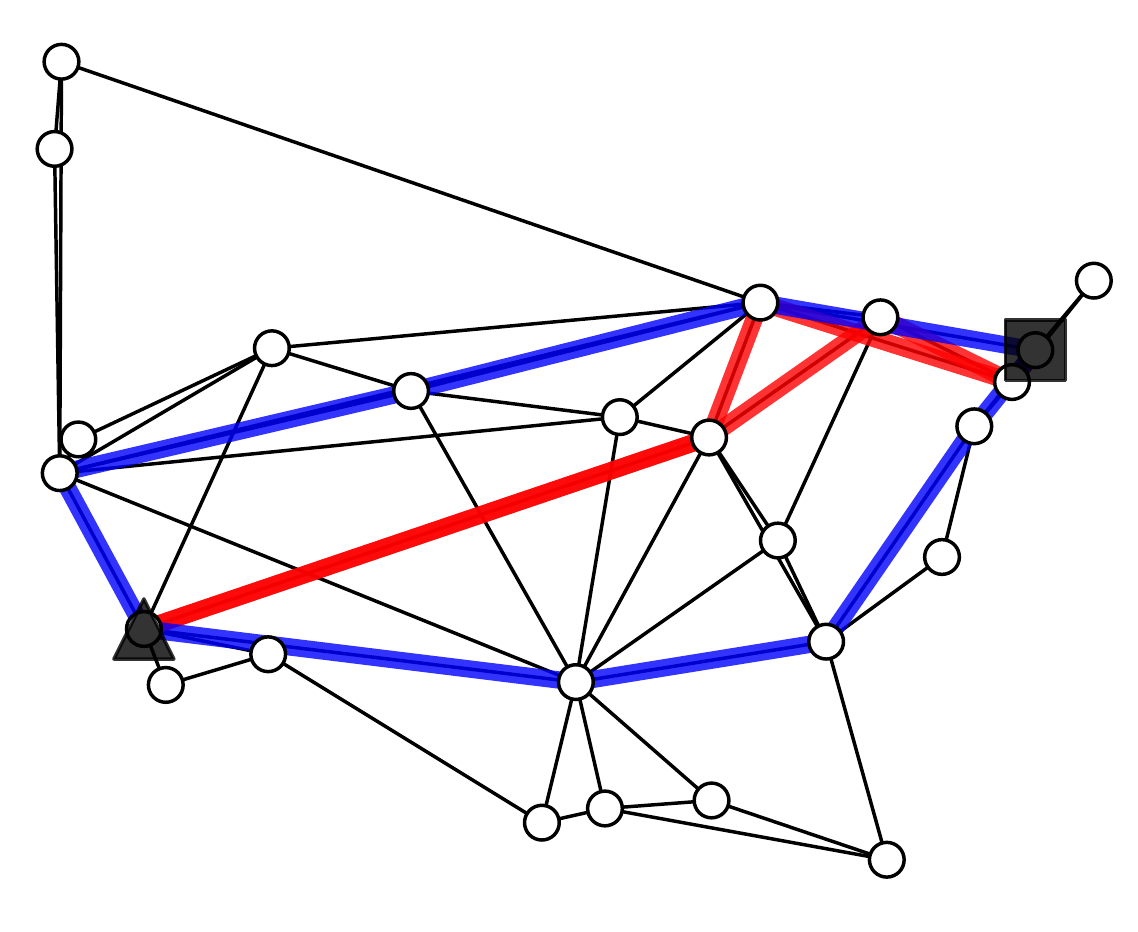}
			\subcaption{ATT}
		\end{center}
	\end{minipage}
\caption{
(a) and (b) show the regret values for the \osp\ and \dst, respectively. 
The regret values are averaged over 100 trials and the error bars indicate the standard deviations.
(c) and (d) show the regret values of each player for the \cg\ on the MCI and ATT, respectively. 
(e) and (f) are the topologies of the two networks; 
the triangles are the starting nodes and the squares are the goal nodes. 
The red (blue) paths indicate the top two paths most frequently chosen by player 1 (2).
}
\label{fig:experiments}
\end{figure}
\fi

\subsection{CG on real-world networks}

\def\nstart{n_{\text{start}}}
\def\nend{n_{\text{end}}}

{\bf Experimental Setting}:
We applied our \combandms\ with ZDDs to the \cg, which is a multi-player version of the \osp, on two real-world networks.
The \cg\ is described as follows:
Given $m$ players and an undirected network with a starting node $s$ and a goal node $t$, the players concurrently send a message from $s$ to $t$.
The aim of each player is to minimize the cumulative time needed to send $T$ messages.
In this problem an arm corresponds to an edge of 
a given network, and a super arm is an $s$-$t$ path. 
The loss value of an arm is the transmission time required when using the edge, 
and the cost of a super arm is the total transmission time needed to send a message along the selected $s$-$t$ path.
In the experiments, we assume that the loss of each edge increases with the number of players who use the same edge at the same time; 
therefore, a player regards the other players as adversaries.
We use $X_t^k \in \Fl$ ($k \in [m]$) to denote the $k$-th player's choice in the $t$-th round and use $X_{t,i}^k \in \{0, 1\}$ to denote the $i$-th element of $\ib_{X_t^k}$.
We also use $\ell_{t,i}^k$ to denote the transmission time that the $k$-th player consumes when sending 
a message using the $i$-th edge 
at the $t$-th round. 
We here define 
$\loss{t}{i}^k := \b_i \kappa^{N_{t,i}^{-k}}$, where $\b_i \in \R$ is the length of the edge, $\kappa$ is an overhead constant, and $N_{t,i}^{-k} := \sum_{k^\prime \not= k} X_{t,i}^{k^\prime}$ is the number of adversaries who also choose the $i$-th edge at the $t$-th round.
Namely, we assume that the transmission time of each edge increases exponentially with the number of players using the same edge at the same time.
Consequently, to reduce the total transmission time, the players should adaptively avoid contending with each other.
Note that this setting violates the assumption $|c_t| < 1$; however, in practice, this violation barely matters.
In the experiments, we set $m = 2$ and $\kappa = 10$.

%

We use two real-world communication networks in the Internet topology zoo~\cite{knight2011internet}: the InternetMCI network (MCI) and the ATT North America network (ATT).
Figure~\ref{fig:experiments} (e) and (f) illustrate the topologies of the MCI and ATT, respectively.
Both networks correspond to the U.S. map and we choose Los Angeles as the starting point $s$ and New York as the goal $t$.
The statistics for each network are shown in Table~\ref{tbl:info}.

\begin{table}[tbh]
\caption{Statistics for two real-world communication networks.}
\label{tbl:info}
\vspace{-10pt}
\begin{small}
\begin{center}
\begin{tabular}{crrrr}
Network & \# nodes & \# edges & \# $s$-$t$ paths $|\Fl|$ &ZDD size $|V|$ \\
\hline
MCI & 19 & 33 & 1,444   & 756    \\
ATT & 25 & 56 & 213,971 & 37,776 \\
\end{tabular}
\end{center}
\end{small}
\end{table}
%

{\bf Experimental Results}:
Figures~\ref{fig:experiments}~(c) and~(d) show the regret values of each player for the MCI and ATT, respectively. 
The figure shows that each player attained sublinear regrets.
Figures~\ref{fig:experiments}~(e) and~(f) show the top two most frequently selected paths for each player. 
We see that each player successfully avoided congestion.
In the {\it full information setting} where the players can observe the costs of all $s$-$t$ paths after choosing the current path, it is known that the Hedge algorithm~\cite{freund1999hedge} can achieve the Nash equilibria~\cite{kleinberg2009congestion} on the \cg. 
In this experiment, 
even though we employed the {\it bandit setting} 
where each player can only observe the cost of the selected path, the players successfully found almost optimal strategies on 
both networks.
To conclude, the experimental results suggest that our algorithm is useful for adaptive routing problems on real-world networks.

\section{Conclusion}
We proposed \combandms\ with ZDDs, which is a practical and theoretically guaranteed algorithm for the adversarial \cmab. 
We also  
showed that our algorithm is effective for network-based 
adversarial \cmab\ instances, 
which include various important problems 
such as the \osp, \dst, and \cg.
The efficiency of our algorithm is thanks to 
the compression of the decision sets via ZDDs, 
and its time and space complexities are 
linear in the size of ZDDs; 
more precisely, they are $O(d|V|)$. 
We showed experimentally that the ZDDs for the \osp, \dst, and \cg\ are much smaller in size than original decision sets.
%
Our algorithm is also theoretically guaranteed to  achieve either $O(T^{2/3})$ regret with high probability or $O(\sqrt{T})$ expected regret as an any-time guarantee;  
we experimentally confirmed that our algorithm attained sublinear regrets.
The results on \cg\ showed that our algorithm is useful for adaptive routing problems on real-world networks. 


	\newpage

	{\small
	\bibliographystyle{plain}
	\bibliography{./mybib}

\begin{thebibliography}{10}

\bibitem{ailon2014permutation}
N.~Ailon, K.~Hatano, and E.~Takimoto.
\newblock Bandit online optimization over the permutahedron.
\newblock In {\em International Conference on Algorithmic Learning Theory},
  pages 215--229. Springer, 2014.

\bibitem{arora2012online}
R.~Arora, O.~Dekel, and A.~Tewari.
\newblock Online bandit learning against an adaptive adversary: from regret to
  policy regret.
\newblock {\em arXiv preprint arXiv:1206.6400}, 2012.

\bibitem{audibert2014regret}
J.-Y. Audibert, S.~Bubeck, and G.~Lugosi.
\newblock Regret in online combinatorial optimization.
\newblock {\em Math. Oper. Res.}, 39(1):31--45, February 2014.

\bibitem{awerbuch2004adaptive}
B.~Awerbuch and R.~D. Kleinberg.
\newblock Adaptive routing with end-to-end feedback: Distributed learning and
  geometric approaches.
\newblock In {\em 36th Annual ACM Symposium on Theory of Computing}, pages
  45--53. ACM, 2004.

\bibitem{bartlett2008high}
P.~L. Bartlett, V.~Dani, T.~Hayes, S.~Kakade, A.~Rakhlin, and A.~Tewari.
\newblock High-probability regret bounds for bandit online linear optimization.
\newblock In {\em 21st Annual Conference on Learning Theory}, pages 335--342.
  Omnipress, 2008.

\bibitem{bergman2016decision}
D.~Bergman, A.~A. Cire, W.-J. van Hoeve, and J.~Hooker.
\newblock {\em Decision Diagrams for Optimization}.
\newblock Springer, first edition, 2016.

\bibitem{braun2016efficient}
G.~Braun and S.~Pokutta.
\newblock An efficient high-probability algorithm for linear bandits.
\newblock {\em arXiv preprint arXiv:1610.02072}, 2016.

\bibitem{cesa2012towards}
S.~Bubeck, N.~Cesa-Bianchi, and S.~M. Kakade.
\newblock Towards minimax policies for online linear optimization with bandit
  feedback.
\newblock In {\em Conference on Learning Theory}, pages 1--14, 2012.

\bibitem{cesa2012comband}
N.~Cesa-Bianchi and G.~Lugosi.
\newblock Combinatorial bandits.
\newblock {\em J. Comput. Syst. Sci.}, 78(5):1404 -- 1422, 2012.

\bibitem{combes2015combinatorial}
R.~Combes, M.~Sadegh Talebi, A.~Proutiere, and M.~Lelarge.
\newblock Combinatorial bandits revisited.
\newblock In {\em Advances in Neural Information Processing Systems}, pages
  2116--2124, 2015.

\bibitem{coudert1997solving}
O.~Coudert.
\newblock Solving graph optimization problems with {ZBDD}s.
\newblock In {\em 1997 European Conference on Design and Test}, page 224. IEEE
  Computer Society, 1997.

\bibitem{fan2012hoeffding}
X.~Fan, I.~Grama, and Q.~Liu.
\newblock Hoeffding's inequality for supermartingales.
\newblock {\em Stoch. Proc. Appl.}, 122(10):3545--3559, 2012.

\bibitem{freund1999hedge}
Y.~Freund and R.~E. Schapire.
\newblock Adaptive game playing using multiplicative weights.
\newblock {\em Games Econom. Behav.}, 29(1):79 -- 103, 1999.

\bibitem{gyorgy2007online}
A.~Gy{\"o}rgy, T.~Linder, G.~Lugosi, and G.~Ottucs{\'a}k.
\newblock The on-line shortest path problem under partial monitoring.
\newblock {\em J. Mach. Learn. Res.}, 8(Oct):2369--2403, 2007.

\bibitem{imase1991dynamic}
M.~Imase and B.~M. Waxman.
\newblock Dynamic {S}teiner tree problem.
\newblock {\em SIAM J. Discrete. Math.}, 4(3):369--384, 1991.

\bibitem{inoue2016graphillion}
T.~Inoue, H.~Iwashita, J.~Kawahara, and S.~Minato.
\newblock Graphillion: software library for very large sets of labeled graphs.
\newblock {\em Int. J. Software Tool. Tech. Tran.}, 18(1):57--66, 2016.

\bibitem{inoue2016acceleration}
Y.~Inoue and S.~Minato.
\newblock Acceleration of {ZDD} construction for subgraph enumeration via
  path-width optimization.
\newblock Technical report, TCS-TR-A-16-80, Hokkaido University, 2016.

\bibitem{Ishihata08}
M.~Ishihata, Y.~Kameya, T.~Sato, and S.~Minato.
\newblock Propositionalizing the {EM} algorithm by {BDD}s.
\newblock In {\em 18th International Conference on Inductive Logic
  Programming}, pages 44--49, 2008.

\bibitem{Ishihata11}
M.~Ishihata and T.~Sato.
\newblock Bayesian inference for statistical abduction using {M}arkov chain
  {M}onte {C}arlo.
\newblock In {\em 3rd Asian Conference on Machine Learning}, pages 81--96,
  2011.

\bibitem{kale2010slate}
S.~Kale, L.~Reyzin, and R.~E. Schapire.
\newblock Non-stochastic bandit slate problems.
\newblock In {\em Advances in Neural Information Processing Systems}, pages
  1054--1062, 2010.

\bibitem{kawahara2014frontier}
J.~Kawahara, T.~Inoue, H.~Iwashita, and S.~Minato.
\newblock Frontier-based search for enumerating all constrained subgraphs with
  compressed representation.
\newblock Technical report, TCS-TR-A-14-76, Hokkaido University, 2014.

\bibitem{kleinberg2009congestion}
R.~Kleinberg, G.~Piliouras, and E.~Tardos.
\newblock Multiplicative updates outperform generic no-regret learning in
  congestion games: Extended abstract.
\newblock In {\em 41st Annual ACM Symposium on Theory of Computing}, pages
  533--542. ACM, 2009.

\bibitem{knight2011internet}
S.~Knight, H.~X. Nguyen, N.~Falkner, R.~Bowden, and M.~Roughan.
\newblock The {I}nternet topology zoo.
\newblock {\em IEEE J. Sel. Area. Comm.}, 29(9):1765--1775, 2011.

\bibitem{knuth2011art1}
D.~E. Knuth.
\newblock {\em The Art of Computer Programming: Combinatorial Algorithms, Part
  1}, volume~4A.
\newblock Addison-Wesley Professional, 1st edition, 2011.

\bibitem{kuleshov2014algorithms}
V.~Kuleshov and D.~Precup.
\newblock Algorithms for multi-armed bandit problems.
\newblock {\em arXiv preprint arXiv:1402.6028}, 2014.

\bibitem{li2010contextual}
L.~Li, W.~Chu, J.~Langford, and R.~E. Schapire.
\newblock A contextual-bandit approach to personalized news article
  recommendation.
\newblock In {\em 19th international conference on World wide web}, pages
  661--670. ACM, 2010.

\bibitem{minato1993zero}
S.~Minato.
\newblock Zero-suppressed {BDD}s for set manipulation in combinatorial
  problems.
\newblock In {\em 30th International Design Automation Conference}, pages
  272--277. ACM, 1993.

\bibitem{morrison2016solving}
D.~R. Morrison, E.~C. Sewell, and S.~H. Jacobson.
\newblock Solving the pricing problem in a branch-and-price algorithm for graph
  coloring using zero-suppressed binary decision diagrams.
\newblock {\em {INFORMS} J. Comput.}, 28(1):67--82, 2016.

\bibitem{robbins1985some}
H.~Robbins.
\newblock Some aspects of the sequential design of experiments.
\newblock In {\em Herbert Robbins Selected Papers}, pages 169--177. Springer,
  1985.

\bibitem{rosenthal1973games}
R.~W. Rosenthal.
\newblock A class of games possessing pure-strategy {N}ash equilibria.
\newblock {\em Internat. J. Game Theory}, 2(1):65--67, 1973.

\bibitem{uchiya2010adversarial}
T.~Uchiya, A.~Nakamura, and M.~Kudo.
\newblock Algorithms for adversarial bandit problems with multiple plays.
\newblock In {\em International Conference on Algorithmic Learning Theory},
  pages 375--389. Springer, 2010.

\end{thebibliography}
	}

	\onecolumn
	\begin{center}
		\textbf{\large Supplementary material}
	\end{center}
	\setcounter{section}{0}
	\setcounter{equation}{0}
	\setcounter{figure}{0}
	\setcounter{table}{0}
	\setcounter{page}{1}
	\makeatletter
	\renewcommand{\thesection}{S\arabic{section}}
	\renewcommand{\theequation}{S\arabic{equation}}
	\renewcommand{\thefigure}{S\arabic{figure}}

In what follows we prove 
Theorem~\ref{thm:highprob} and Theorem~\ref{thm:expectation}. 
Section~\ref{supp:1} presents 
two concentration inequalities that are 
important in the proofs. 
In Section~\ref{supp:2} we provide some preliminaries 
for the proofs. 
Section~\ref{supp:3} and Section~\ref{supp:4} provide
the proofs of
Theorem~\ref{thm:highprob} and Theorem~\ref{thm:expectation}, 
respectively. 

\section{Concentration inequalities}\label{supp:1}
	The following concentration inequalities play crucial roles
	in the subsequent discussion. 

	\begin{thm}[Azuma-Hoeffding inequality]\label{thm:azuma}
		If a martingale difference sequence $\{Z_t\}_{t=1}^T$ satisfies
		$a_t\le Z_t \le b_t$ almost surely with some constants $a_t,b_t$ for $t=1,\dots,T$,
		then the following inequality holds with probability at least $1-\delta$:
		\[
		\sum_{t=1}^T Z_t\le \sqrt{\frac{\ln(1/\delta)}{2}\sum_{t=1}^T(b_t-a_t)^2}.
		\]
	\end{thm}

	\begin{thm}[Bennett's inequality~\cite{fan2012hoeffding}]\label{thm:bennett}
		If a supermartingale difference sequence $\{Z_t\}_{t=1}^T$
        with respect to a filtration $\{\mathcal{F}_t\}_{t=0}^{T-1}$ satisfies
		$Z_t\le b$ with some constant $b>0$ 
        for $t=1,\dots,T$,
		then, for any $v\ge0$, we have the following with probability at least $1-\delta$:
		\[
		\sum_{t=1}^T\var{Z_t \mid \mathcal{F}_{t-1}}\ge v \quad\text{or}\quad \sum_{t=1}^TZ_t \le \frac{b}{3}\ln\frac{1}{\delta} + \sqrt{2v\ln\frac{1}{\delta}}.
		\]
    	\end{thm}

\section{Preliminaries for the proofs}\label{supp:2}

	\begin{algorithm}
	\caption{\combandms$(\alpha, \Fl)$}\label{alg:combandalt}
	\begin{algorithmic}[1]
		\STATE $\tlweight{1}{i}\gets1$ and $\weight{1}{i}\gets1$ ($i\in E$)
		\FOR{$t=1,\dots,T$}
		\STATE $\gamma_t\gets \frac{t^{-1/\a}}{2}$, $\eta_t\gets\frac{\lambda t^{-1/\a}}{2\bdx^2}$, $\eta_{t+1}\gets\frac{\lambda (t+1)^{-1/\a}}{2\bdx^2}$
		\STATE $X_t \sim p_t$
		\STATE $c_t\gets\lb_t^\top\ib_{\Xt}$ ($\lb_t$ is unobservable)
			\STATE $P_t\gets(1-\gamma_t)Q_t + \gamma_t U$
			\STATE $\lbh_t\gets c_tP_t^+\ib_{\Xt}$
			\STATE $\weight{t+1}{i}\gets \tlweight{t}{i}\exp\big(-\eta_t\hloss{t}{i}\big)$ $(i\in E)$
			\STATE $\tlweight{t+1}{i}\gets \weight{t+1}{i}^{{\eta_{t+1}}/{\eta_{t}}}$ $(i\in E)$
			\ENDFOR
			\RETURN $\{X_t \mid t \in [T]\}$
		\end{algorithmic}
	\end{algorithm}

    We here rewrite Algorithm~\ref{alg:comband} equivalently as in Algorithm~\ref{alg:combandalt},
    which will be helpful in terms of understanding the subsequent discussion.
	In what follows, we let $K:=|\Fl|$ and $\mu:=1/K$.
    We also define
	$\E_t[\cdot]:=\E[\cdot\mid \Xsub{1:t-1}, \lsub{1:t} ]$
	as the conditional expectation in the $t$-th round
	given all the history of rounds $1,\dots,t-1$
	and the loss vector in round $t$.
	Similarly, we define the conditional variance in
	round $t$ as
	$\vart{\cdot}:=
	\var{\cdot\mid \Xsub{1:t-1},\lsub{1:t}}$.
    For any vector $\xb=(x_1,\dots,x_n)^\top\in\R^n$ and $p>0$, we define the $p$-norm of $\xb$ as
$\norm{\xb}_p:=(\sum_{i=1}^n |x_i|^p)^{1/p}$,
and we often use $\norm{\xb}$ to express $\norm{\xb}_2$.
	For any matrix $P\in\R^{n\times n}$, 
	we denote its $i,j$ entry as $P(i,j)$. 
	We define the trace of $P$ as  $\tr(P):=\sum_{i=1}^n P(i,i)$ 
	and denote the spectral norm of $P$ as $\norm{P}$,
	i.e., $\norm{P}$
	is the largest singular value of $P$. 
    For any symmetric 
	matrices $P,Q\in\R^{n\times n}$, 
	we use $P\succeq Q$ to express the fact that
	the smallest eigenvalue of $P-Q$ is non-negative.
    
	For all $t\in[T]$, 
	we define distributions $u$ and $q_t$ over $\Fl$, 
    and $d\times d$ matrices $U$ and $Q_t$ as follows:
 	\begin{align*}
 	&u(X):=p(X;\ib_E, \Fl)=\mu, 
     & &
     q_t(X):=p(X;\wbt_t, \Fl)=\frac{\wt_t(X)}
     {\sum_{X'\in\Fl}\wt_t(X') },
     \\
     &U:=\E_{X\sim u}
     [\ib_X\ib_X^\top]=\sum_{X\in\Fl}\mu\ib_X\ib_X^\top,
    & & 
    Q_t:=\E_{X\sim q_t}
    [\ib_X\ib_X^\top]=\sum_{X\in\Fl}q_t(X)\ib_X\ib_X^\top,
 	\end{align*}
where $\wt_t(X)$ is an abbreviation of $\prod_{i\in X}\wt_{t,i}$. 
	Note that we have the following for any $X\in\Fl$ and $t\in[T]$:
	\begin{align*}
	&p_t(X)=(1-\gamma_t)q_t(X)+\gamma_t u(X), \\
	&P_t=\E_{X\sim p_t}[\ib_X\ib_X^\top]=\sum_{X\in\Fl}p_t(X)\ib_X\ib_X^\top=(1-\gamma_t)Q_t+\gamma_tU,
	\end{align*}
	where $p_t(X)$ and $P_t$ are those defined in 
    Eq.~\eqref{def:dists} and \eqref{eq:cmppt}, respectively. 
    We note that the weight values $w_{t,i}\ (i\in E)$ defined 
    in Step~8 of Algorithm~\ref{alg:combandalt} 
    satisfy the following for any $X\in\Fl$ and $t\ge2$:
	\begin{align}
	w_t(X)&=
	\exp\bigg(-\eta_{t-1}\sum_{t'=1}^{t-1}\lbh_{t'}^\top\ib_X\bigg),
	\label{eq:wtsum}
	\\
	q_t(X)&=
	\frac{w_t(X)^{\frac{\eta_t}{\eta_{t-1}}}}
	{\sum_{X'\in\Fl} w_t(X')^{\frac{\eta_t}{\eta_{t-1}}}}.
	\label{eq:qtsum}
	\end{align}
	For convenience, we let $\eta_0:=\eta_1$  in what follows,
	which makes Eq.~\eqref{eq:qtsum} hold
	for $t=1$ since we have $\weight{1}{i}=\tlweight{1}{i}=1$ 
    for all $i\in E$. 
    
    Recall that $\lambda$ is the smallest non-zero eigenvalue 
of $U=\E_{X\sim u}[\ib_X\ib_X^\top]$, 
and that $|c_t|\le1$ holds because of the loss value assumption. 
	The following basic results will be used repetitively in what follows.
	\begin{lem}[Basic results]\label{lem:basic}
		For any $X\in\Fl$ and $t\in[T]$, we have
		\begin{align}
		&\norm{P_t^+}\le\frac{1}{\gamma_t \lambda}
		\quad\text{and}\quad
		|\lbh_t^\top\ib_X|\le\frac{\bdx^2}{\gamma_t\lambda}, \label{basic1}\\
		&\E_t[　\ib_{\Xt}^\top P_t^+ \ib_{\Xt}]
		\le d, \label{basic2}\\
		&P_tP_t^+\ib_X=\ib_X, \label{basic3}\\
		&\E_t[　\lbh_t^\top \ib_X]=\lb_t^\top\ib_X.  \label{basic4}
		\end{align}
	\end{lem}

	\begin{proof}
		The first inequality of Eq.~\eqref{basic1} comes from $P_t\succeq \gamma_t U$,
		and the second one is obtained from $|c_t|\le1$ as follows:
		\[
		|\lbh_t^\top\ib_X|=
		|c_t \ib_{\Xt}^\top P_t^+\ib_{X}|
		\le
		\norm{\ib_{\Xt}} \norm{P_t^+} \norm{\ib_{X}}
		\le
		\frac{\bdx^2}{\gamma_t\lambda}.
		\]
		Eq.~\eqref{basic2} can be obtained as follows:
		\[
		\E_t[　\ib_{\Xt}^\top P_t^+ \ib_{\Xt}]
		=
		\E_t[　\tr( P_t^+ \ib_{\Xt}\ib_{\Xt}^\top)]
		=
		\tr(P_t^+P_t)
		\le d.
		\]
		The proof of Eq.~\eqref{basic3} is
		presented in~\cite[Lemma 14]{cesa2012comband}.
		Finally, Eq.~\eqref{basic4} is obtained with Eq.~\eqref{basic3}
		as follows:
		\[
		\E_t[　\lbh_t^\top \ib_X]=
		\E_t[\lb_t^\top\ib_\Xt\ib_{\Xt}P_t^+\ib_X]=
		\lb_t^\top P_tP_t^+\ib_X=
		\lb_t^\top\ib_X.
		\]
	\end{proof}

	\section{Proof for the high-probability regret bound}\label{supp:3}
	We show the complete proof of Theorem~\ref{thm:highprob}.
	Below is a detailed statement of the theorem.

	\begin{thm}\label{thm:highprobdetail}
		The sequence of super arms $\{\Xt\}_{t\in[T]}$
        obtained by
        \combandms$(\a=3,\Fl)$ satisfies the following
        inequality
        for any $X\in\Fl$ with probability at least $1-\delta$:
		\[
		\sum_{t=1}^T (\lb_t^\top \ib_\Xt - \lb_t^\top \ib_X)
		\le
		\bigg(\frac{3d(e-2)\lambda}{4\bdx^2}+\frac{3}{2}+\bdx\sqrt{\frac{7}{\lambda}\ln \frac{K+2}{\delta}}\bigg)T^{2/3} + o(T^{2/3}).
		\]
	\end{thm}

	Let $\xbt_t:=\sum_{X\in\Fl} q_t(X)\ib_{X}$.
	As in~\cite{braun2016efficient}, the proof
	is obtained by bounding each term on the right hand side of the following equation:
	\begin{align*}
	\sum_{t=1}^T (\lb_t^\top \ib_\Xt - \lb_t^\top \ib_X)
	=\sum_{t=1}^T (\lb_t^\top \ib_\Xt-\lbh_t^\top \xbt_t)
	+ \sum_{t=1}^T (\lbh_t^\top \xbt_t-\lbh_t^\top \ib_X)
	+ \sum_{t=1}^T (\lbh_t^\top \ib_X - \lb_t^\top\ib_X),
	\end{align*}
	where $X\in\Fl$ is an arbitrary super arm.
	To bound them,
	we prove the following three lemmas.


	\begin{lem}\label{lem:highprob1}
		For any $X\in\Fl$, we have
		\begin{align*}
		\sum_{t=1}^T
		(\lbh_t^\top\xbt_t
		-\lbh_t^\top\ib_{X})
		\le
		\frac{\ln K}{\eta_T}
		+
		(e-2)
		\Bigg(
		d\sum_{t=1}^T \frac{ \eta_t }{1-\gamma_t} +
		\frac{\bdx^2}{\lambda} \sqrt{\frac{1}{2}\ln\frac{1}{\delta}\sum_{t=1}^T\frac{\eta_{t}^2}{\gamma_t^2(1-\gamma_t)^2} }
		\Bigg)
		\end{align*}
		with probability at least $1-\delta$.
	\end{lem}

	\begin{proof}
    With the weight values 
    $\weight{t}{i}\ (i\in E)$ used in 
    Algorithm~\ref{alg:combandalt}, 
    we define $w_t(X):=\prod_{i\in X}w_{t,i}$ and 
    $W_t:=\sum_{X\in\Fl} w_t(X)$;  
	we measure the progress of the algorithm in each round 
    via $\ln({W_{t+1}^{\eta_t^{-1}}}/{W_{t}^{\eta_{t-
    1}^{-1}}})$.
		By H{\"o}lder's inequality,
		$\norm{\xb}_s\ge
		K^{\frac{1}{s}-\frac{1}{r}}\norm{\xb}_r$ holds
		for any $\xb\in\R^K$ and $0<r\le s$.
		Thus, letting $s=\eta_{t-1}/\eta_{t}$
		and $r=1$,
		we obtain
		\begin{align*}
		W_{t}^{\eta_{t-1}^{-1}}
		&=
		\bigg(\sum_{X\in\Fl}w_t(X)^{\frac{\eta_t}{\eta_{t-1}}\frac{\eta_{t-1}}{\eta_{t}}}\bigg)^{\frac{\eta_t}{\eta_{t-1}}{\eta_t^{-1}}}
		\ge
		K^{\frac{1}{\eta_{t-1}}-\frac{1}{\eta_{t}}}
		\bigg(\sum_{X\in\Fl}w_t(X)^{\frac{\eta_t}{\eta_{t-1}}}\bigg)^{\eta_{t}^{-1}}.
		\end{align*}
		Hence we have
		\begin{align*}
		\ln \frac{W_{t+1}^{\eta_{t}^{-1}}}{W_{t}^{\eta_{t-1}^{-1}}}
		-\ln K^{\frac{1}{\eta_{t}}-\frac{1}{\eta_{t-1}}}
		&\le
		\frac{1}{\eta_{t}}\ln \frac{W_{t+1}}{\sum_{X\in\Fl}w_t(X)^{\frac{\eta_t}{\eta_{t-1}}}}\\
		&=
		\frac{1}{\eta_{t}}\ln\sum_{X\in\Fl}\frac{ w_{t}(X)^{\frac{\eta_t}{\eta_{t-1}}} \exp(-\eta_t\lbh_t^\top\ib_{X})}{\sum_{X'\in\Fl}w_t(X')^{\frac{\eta_t}{\eta_{t-1}}}} \\
		&=
		\frac{1}{\eta_t}\ln
		\sum_{X\in\Fl}
		q_{t}(X)\exp(-\eta_t\lbh_t^\top\ib_{X}) \\
		&\le
		\frac{1}{\eta_t}\ln
		\sum_{X\in\Fl}q_{t}(X)\left(1-\eta_t\lbh_t^\top\ib_{X}+
		(e-2)\eta_t^2(\lbh_t^\top\ib_{X})^2 \right)
		\\
		&=
		\frac{1}{\eta_{t}}\ln
		\left(1-\eta_t\lbh_t^\top\xbt_t+
		(e-2)\eta_t^2\sum_{X\in\Fl} q_t(X)
		(\lbh_t^\top\ib_{X})^2\right) \\
		&\le
		-\lbh_t^\top\xbt_t+
		(e-2)\eta_t\sum_{X\in\Fl} q_t(X)
		(\lbh_t^\top\ib_{X})^2,
		\end{align*}
		where the second inequality comes from
		$e^{-x} \le 1-x+(e-2)x^2$ for any $|x|\le1$;
		note that $\eta_t$ is defined to satisfy
		$\eta_t|\lbh_t^\top\ib_{X}|\le \eta_t \bdx^2/(\gamma_t\lambda)=1$.
		The third inequality is obtained by $\ln(1+x)\le x$ for any $x\ge-1$.
		The second term on the right hand side is bounded
        from above as follows:
		\begin{align*}
		\sum_{X\in\Fl}q_{t}(X)(\lbh_t^\top\ib_{X})^2
		&\le \sum_{X\in\Fl}\frac{p_{t}(X)}{1-\gamma_t}(\lbh_t^\top\ib_{X})^2
		\le \frac{\ib_{\Xt}^\top P_t^+ \ib_{\Xt}}{1-\gamma_t}.
		\end{align*}
        Therefore, we have
		\begin{align*}
        \frac{1}{\eta_t}
        \ln W_{t+1} -
        \frac{1}{\eta_{t-1}}
        \ln W_{t}
		&\le
        \bigg(\frac{1}{\eta_t}-\frac{1}{\eta_{t-1}}\bigg)
        \ln K
		-\lbh_t^\top\xbt_t+
		(e-2)\eta_t
        \frac{\ib_{\Xt}^\top P_t^+ \ib_{\Xt}}{1-\gamma_t}.
		\end{align*}
		Summing up both sides of the above 
        for $t=1,\dots,T$, we obtain
		the following inequality from $W_1=K$:
		\begin{align*}
		\frac{1}{\eta_T}\ln W_{T+1}
		\le
        \frac{1}{\eta_T}\ln K
		-\sum_{t=1}^T \lbh_t^\top\xbt_t+
		(e-2)
		\sum_{t=1}^T \eta_t\frac{\ib_{\Xt}^\top P_t^+ \ib_{\Xt}}{1-\gamma_t}.
		\end{align*}

		On the other hand,
		we have $\weight{T+1}{i} =
		\exp\big(-\eta_T \sum_{t=1}^T \hloss{t}{i}\big)$ by Eq.~\eqref{eq:wtsum}.
		Thus the following holds for any $X\in\Fl$:
		\begin{align*}
		\frac{1}{\eta_T}\ln W_{T+1}
		\ge \frac{1}{\eta_T} \ln w_{T+1}(X)
		= -\sum_{t=1}^T \lbh_t^\top\ib_{X}.
		\end{align*}
		Therefore, we obtain
		\begin{equation}\label{ineq:lemtwokey}
		\sum_{t=1}^T
		(\lbh_t^\top\xbt_t
		- \lbh_t^\top\ib_{X})
		\le
		\frac{\ln K}{\eta_T}
		+
		(e-2)
		\sum_{t=1}^T \eta_t\frac{\ib_{\Xt}^\top P_t^+ \ib_{\Xt}}{1-\gamma_t}.
		\end{equation}
		The second term on the right hand side
		can be bounded from above
		by using the Azuma--Hoeffding inequality~(Theorem~\ref{thm:azuma}) for
		the martingale difference sequence
		$\frac{\eta_t}{1-\gamma_t}
		(\ib_{\Xt}^\top P_t^+ \ib_{\Xt}-\E_t[\ib_{\Xt}^\top P_t^+ \ib_{\Xt}])$ as follows.
		First, note that we have
		\[
		\E_t[　\ib_{\Xt}^\top P_t^+ \ib_{\Xt}]
		\le d
		\quad \text{and} \quad
		0\le \frac{\eta_t \ib_{\Xt}^\top P_t^+ \ib_{\Xt}}{1-\gamma_t}
		\le \frac{\eta_t\bdx^2}{(1-\gamma_t)\gamma_t\lambda}
		\]
		by Lemma~\ref{lem:basic}.
		Thus, by the Azuma-Hoeffding inequality,
		the following holds
		with probability at least $1-\delta$:
		\[
		\sum_{t=1}^T \frac{\eta_t \ib_{\Xt}^\top P_t^+ \ib_{\Xt}}{1-\gamma_t}
		\le
		d\sum_{t=1}^T \frac{ \eta_t }{1-\gamma_t} +
		\frac{\bdx^2}{\lambda} \sqrt{\frac{\ln(1/\delta)}{2}\sum_{t=1}^T\frac{\eta_{t}^2}{\gamma_t^2(1-\gamma_t)^2} }.
		\]
		Hence we obtain
		\[
		\sum_{t=1}^T
		(\lbh_t^\top\xbt_t
		- \lbh_t^\top\ib_{X})
		\le
		\frac{\ln K}{\eta_T}
		+
		(e-2)
		\\
		\Bigg(
		d\sum_{t=1}^T \frac{ \eta_t }{1-\gamma_t} +
		\frac{\bdx^2}{\lambda} \sqrt{\frac{\ln(1/\delta)}{2}\sum_{t=1}^T\frac{\eta_{t}^2}{\gamma_t^2(1-\gamma_t)^2} }
		\Bigg).
		\]
	\end{proof}


	\begin{lem}\label{lem:highprob2}
    The following inequality holds with probability
    at least $1-\delta$:
		\begin{align*}
		\sum_{t=1}^T(\lb_t^\top\ib_{\Xt}-\lbh_t^\top\xbt_t) 
        \le &\ 
		2\sum_{t=1}^T\gamma_t +\frac{1}{3}\bigg(2+ \frac{\bdx}{\sqrt{\lambda\gamma_T(1-\gamma_T)}}\bigg)\ln\frac{1}{\delta}
        \\
        &+\sqrt{2\bigg(T+\frac{3\bdx^2}{\lambda}\sum_{t=1}^T\frac{\gamma_t}{(1-\gamma_t)^2}\bigg)\ln\frac{1}{\delta} }.
		\end{align*}
	\end{lem}
    \begin{proof}
		Let $\zb:=\sum_{X\in\Fl}\mu \ib_X$ and
		$\xbb_t:=\E_t[\ib_{\Xt}]=(1-\gamma_t)\xbt_t +\gamma_t \zb$.
		We obtain the proof by 
		using 
		Bennett's inequality~(Theorem~\ref{thm:bennett}) for the martingale difference sequence
		\begin{align*}
		Y_t
		:=&\
		\lb_t^\top \ib_{\Xt} - \lbh_t^\top \xbt_t
		-
		\E_t[\lb_t^\top \ib_{\Xt} - \lbh_t^\top \xbt_t] \\
		=&\
		\lb_t^\top \ib_{\Xt} - \lbh_t^\top \xbt_t
		-
		\lb_t^\top \xbb_t + \lb_t^\top \xbt_t \\
		=&\
		\lb_t^\top \ib_{\Xt} - \lbh_t^\top \xbt_t
		+
		\gamma_t\lb_t^\top(\xbt_t-\zb).
		\end{align*}
        We first bound the values of 
        $|Y_t|$ and $\vart{Y_t}$. 
		By
        $Q_t\preceq \frac{1}{1-\gamma_t}P_t$ and Jensen's inequality
		$\xbt_t\xbt_t^\top\preceq Q_t$,
        we have
		\[
		(\lbh_t^\top \xbt_t)^2
		\le
		c_t^2\ib_{\Xt}^\top P_t^+ Q_t P_t^+ \ib_{\Xt}
		\le
		\frac{\ib_{\Xt}^\top P_t^+ \ib_{\Xt}}{1-\gamma_t}
		\le
		\frac{\bdx^2}{\lambda\gamma_t(1-\gamma_t)},
		\]
		and hence
		\[
		|Y_t|
		\le 1+ \frac{\bdx}{\sqrt{\lambda\gamma_t(1-\gamma_t)}} + 2\gamma_t
		\le 2+ \frac{\bdx}{\sqrt{\lambda\gamma_T(1-\gamma_T)}}.
		\]
		The variance of $Y_t$ is bounded as follows:
		\begin{align*}
		\vart{Y_t}
		&\le
		\E_t[(\lb_t^\top \ib_{\Xt} - \lbh_t^\top \xbt_t)^2]
		=
		\E_t[c_t^2(1 - \ib_{\Xt}^\top P^+_t \xbt_t)^2]
		\le
		\E_t[(1 - \ib_{\Xt}^\top P^+_t \xbt_t)^2] \\
		&=
		\E_t[1 - 2 \ib_{\Xt}^\top P^+_t \xbt_t +
		\xbt_t^\top P^+_t \ib_{\Xt}
		\ib_{\Xt}^\top P^+_t \xbt_t] \\
		&=
		1 - 2 \xbb_t^\top P^+_t \xbt_t +
		\xbt_t^\top P^+_t  \xbt_t \\
		&=
		1 - \frac{2}{1-\gamma_t} \xbb_t^\top P^+_t (\xbb_t-\gamma_t\zb) +
		\frac{1}{(1-\gamma_t)^2}
		(\xbb_t-\gamma_t\zb)^\top P^+_t  (\xbb_t-\gamma_t\zb) \\
		&=
		1 - \frac{1-2\gamma_t}{(1-\gamma_t)^2} \xbb_t^\top P^+_t \xbb_t
		+ \frac{\gamma_t^2}{(1-\gamma_t)^2}
		(\zb - 2\xbb_t)^\top P^+_t  \zb \\
		&\le
        1 + \frac{\gamma_t^2}{(1-\gamma_t)^2}
		(\zb - 2\xbb_t)^\top P^+_t  \zb
         \\
		&\le
		1 + \frac{3\gamma_t \bdx^2}{(1-\gamma_t)^2\lambda},
		\end{align*}
		where the third inequality comes from
        $1-2\gamma_t\ge0$, and the last inequality is
        obtained by Lemma~\ref{lem:basic} with $\norm{\zb}\le L$ and
        $\norm{\xbb_t}\le L$.
		Therefore, by using Bennett's inequality,
		we obtain
		\begin{align*}
		\sum_{t=1}^T
		(\lb_t^\top \ib_{\Xt} - \lbh_t^\top \xbt_t
		+
		\gamma_t\lb_t^\top(\xbt_t-\zb))
		\le &\ 
		\frac{1}{3}\bigg(2+ \frac{\bdx}{\sqrt{\lambda\gamma_T(1-\gamma_T)}}\bigg)\ln\frac{1}{\delta}
        \\
        &+\sqrt{2\bigg(T+\frac{3\bdx^2}{\lambda}\sum_{t=1}^T\frac{\gamma_t}{(1-\gamma_t)^2}\bigg)\ln\frac{1}{\delta}}.
		\end{align*}
		The proof is completed by $|\lb_t^\top(\xbt_t-\zb)|\le 2$.
	\end{proof}

	\begin{lem}\label{lem:highprob3}
    The following inequality holds for
    all $X\in\Fl$ simultaneously with probability
    $1-\delta$:
		\begin{align*}
		\sum_{t=1}^T(\lbh_t^\top\ib_{X}-\lb_t^\top\ib_{X})\le
		\frac{1}{3}\bigg(1+\frac{\bdx^2}{\gamma_T \lambda} \bigg)\ln\frac{K}{\delta} +
		\sqrt{\frac{2\bdx^2}{\lambda}\ln\frac{K}{\delta}\sum_{t=1}^T\frac{1}{\gamma_t}  }.
		\end{align*}
	\end{lem}

	\begin{proof}
		We fix $X\in\Fl$ arbitrarily.
		The proof is obtained by
		using Bennett's inequality for
		the martingale difference sequence $\lbh_t^\top \ib_{X} - \lb_t^\top \ib_{X}$; note that
        $\E_t[\lbh_t^\top \ib_{X} - \lb_t^\top \ib_{X}]=0$ holds by Lemma~\ref{lem:basic}.
		First, the absolute value and variance of
        $\lbh_t^\top \ib_{X} - \lb_t^\top \ib_{X}$
        are bounded as follows:
        \begin{align*}
        &|\lbh_t^\top \ib_{X} - \lb_t^\top \ib_{X}|
        \le 1+\frac{\bdx^2}{\gamma_t\lambda},
        \\
        &\vart{\lbh_t^\top \ib_{X} - \lb_t^\top \ib_{X}}
		\le
		\E_t[(\lbh_t^\top\ib_{X})^2]\le
		\E_t[\ib_X^\top P^+_t \ib_{\Xt} \ib_{\Xt}^\top P_t^+ \ib_{X}]
		\le \ib_X P_t^+ \ib_{X} \le \frac{\bdx^2}{\gamma_t\lambda}.
        \end{align*}
		Hence, by Bennett's inequality,
		we have
		\begin{equation*}
		\sum_{t=1}^T(\lbh_t^\top\ib_{X}-\lb_t^\top\ib_{X})\le
		\frac{1}{3}\bigg(1+\frac{\bdx^2}{\gamma_T \lambda} \bigg)\ln\frac{K}{\delta} +
		\sqrt{\frac{2\bdx^2}{\lambda}\ln\frac{K}{\delta}\sum_{t=1}^T\frac{1}{\gamma_t}  }
		\end{equation*}
		with probability at least $1-\delta/K$.
		Taking the union bound over all super arms $X\in\Fl$,
		we obtain the claim.
	\end{proof}

	Using the above three lemmas,
	we prove Theorem~\ref{thm:highprobdetail} as follows.

	\begin{proof}[Proof of Theorem~\ref{thm:highprobdetail}]
		Note that we have
		$\gamma_t=\frac{t^{-1/3}}{2}$
		and
		$\eta_t=\frac{\lambda}{\bdx^2}\gamma_t=
		\frac{\lambda t^{-1/3}}{2\bdx^2}$.
		By using Lemma~\ref{lem:highprob1}, we have the following with probability at least
		$1-\delta/(K+2)$:
		\begin{align*}
		\sum_{t=1}^T
		(\lbh_t^\top\xbt_t
		-\lbh_t^\top\ib_{X})
		&\le
		\frac{\ln K}{\eta_T}
		+
		(e-2)
		\Bigg(
		d\sum_{t=1}^T \frac{ \eta_t }{1-\gamma_t} +
		\frac{\bdx^2}{\lambda} \sqrt{\frac{1}{2}\ln\frac{K+2}{\delta}\sum_{t=1}^T\frac{\eta_{t}^2}{\gamma_t^2(1-\gamma_t)^2} }
		\Bigg) \\
		&\le
		\frac{2\bdx^2\ln K}{\lambda}T^{1/3}
		+
		(e-2)
		\Bigg(
		\frac{3d\lambda}{4\bdx^2}(T^{2/3}+2T^{1/3}) +
		\sqrt{2T\ln\frac{K+2}{\delta}}
		\Bigg).
		\end{align*}
		We also obtain the following inequality
        with probability at least
		$1-\delta/(K+2)$ by using Lemma~\ref{lem:highprob2}:
		\begin{align*}
		&\sum_{t=1}^T(\lb_t^\top\ib_{\Xt}-\lbh_t^\top\xbt_t)\\
		&\le 2\sum_{t=1}^T\gamma_t +\frac{1}{3}\bigg(2+ \frac{\bdx}{\sqrt{\lambda\gamma_T(1-\gamma_T)}}\bigg)\ln\frac{K+2}{\delta}
		+\sqrt{2\bigg(T+\frac{3\bdx^2}{\lambda}\sum_{t=1}^T\frac{\gamma_t}{(1-\gamma_t)^2}\bigg)\ln\frac{K+2}{\delta} } \\
		&\le \frac{3}{2}T^{2/3} +
		\frac{1}{3}\bigg(2+ \frac{\sqrt{2}\bdx}{\sqrt{\lambda}}(T^{1/6}+T^{-1/6})\bigg)\ln\frac{K+2}{\delta}
		+\sqrt{2\bigg(T+\frac{9\bdx^2}{\lambda}T^{2/3}\bigg)\ln\frac{K+2}{\delta} }.
		\end{align*}
		Furthermore, we have the following inequality
        with probability at least $1-K\delta/(K+2)$ by using Lemma~\ref{lem:highprob3}:
		\begin{align*}
		\sum_{t=1}^T(\lbh_t^\top\ib_{X}-\lb_t^\top\ib_{X})
		&\le
		\frac{1}{3}\bigg(1+\frac{\bdx^2}{\gamma_T \lambda} \bigg)\ln\frac{K+2}{\delta} +
		\sqrt{\frac{2\bdx^2}{\lambda}\ln\frac{K+2}{\delta}\sum_{t=1}^T\frac{1}{\gamma_t}  } \\
		&\le
		\frac{1}{3}\bigg(1+\frac{2\bdx^2}{\lambda}T^{1/3} \bigg)\ln\frac{K+2}{\delta} +
		\frac{\sqrt{3}\bdx}{\sqrt{\lambda}}T^{2/3}\sqrt{\bigg(1+\frac{4}{3T}\bigg)\ln\frac{K+2}{\delta}}.
		\end{align*}
		Summing up both sides of the three inequalities
		and taking the union bound,
		we obtain the theorem.
	\end{proof}

	\section{Proof for the expected regret bound}\label{supp:4}

	We then show the proof of Theorem~\ref{thm:expectation}; the detailed statement is as follows.

	\begin{thm}\label{thm:expectdetail}
		The sequence of super arms
        $\{\Xt\}_{t\in[T]}$ obtained by
        \combandms$(\a=2,\Fl)$ satisfies the following inequality for any $X\in\Fl$:
		\begin{align*}
		\E\Bigg[\sum_{t=1}^T (\lb_t^\top \ib_\Xt - \lb_t^\top \ib_X)\Bigg]
		\le
		\bigg(\frac{2\bdx^2\ln K}{\lambda}
		+
		\frac{(e-2)d\lambda}{\bdx^2}
		+
		2\bigg)\sqrt{T} + o(\sqrt{T})
		.
		\end{align*}
	\end{thm}

    Let $\xbt_t:=\sum_{X\in\Fl}q_t(X)\ib_{X}$.
	The proof is obtained by
	bounding each term on the right hand side of the following equation for any $X\in\Fl$:
	\begin{align}
	\E\Bigg[\sum_{t=1}^T (\lb_t^\top \ib_\Xt - \lb_t^\top \ib_X)\Bigg]
	&=
	\E\Bigg[\sum_{t=1}^T \E_t[\lb_t^\top \ib_\Xt - \lb_t^\top \ib_X]\Bigg] \label{ineq:expectterms}
	\\
	&=\E\Bigg[\sum_{t=1}^T \E_t[\lb_t^\top \ib_\Xt-\lbh_t^\top \xbt_t]\Bigg]
	+\E\Bigg[ \sum_{t=1}^T \E_t[\lbh_t^\top \xbt_t-\lbh_t^\top \ib_X]\Bigg], \nonumber
	\end{align}
	where the second equality comes
	from Lemma~\ref{lem:basic}.
	To bound these terms,
	we prove the following two lemmas.


	\begin{lem}\label{lem:expect1}
		For any $X\in\Fl$, we have
		\begin{align*}
		\E\Bigg[\sum_{t=1}^T\E_t[\lbh_t^\top\xbt_t
		-\lbh_t^\top\ib_{X}]\Bigg]
		\le
		\frac{\ln K}{\eta_T}
		+
		(e-2)d\sum_{t=1}^T
		\frac{\eta_t}{1-\gamma_t}.
		\end{align*}
	\end{lem}

	\begin{proof}
		As in the proof of Lemma~\ref{lem:highprob1}, 
        we have Eq.~\eqref{ineq:lemtwokey}; 
		\begin{align*}
		\sum_{t=1}^T
		(\lbh_t^\top\xbt_t
		- \lbh_t^\top\ib_{X})
		\le
		\frac{\ln K}{\eta_T}
		+
		(e-2)
		\sum_{t=1}^T \eta_t\frac{\ib_{\Xt}^\top P_t^+ \ib_{\Xt}}{1-\gamma_t}.
		\end{align*}
		Taking the expectation of both sides, we obtain
		\begin{align*}
		\E\Bigg[\sum_{t=1}^T
		\E_t[\lbh_t^\top\xbt_t
		-\lbh_t^\top\ib_{X}]\Bigg]
		&\le
		\frac{\ln K}{\eta_T}
		+
		(e-2)
		\E\Bigg[\sum_{t=1}^T
		\frac{\eta_t}{1-\gamma_t}\E_t[\ib_{\Xt}^\top P_t^+ \ib_{\Xt}]\Bigg]\\
		&\le
		\frac{\ln K}{\eta_T}
		+
		(e-2)d\sum_{t=1}^T
		\frac{\eta_t}{1-\gamma_t},
		\end{align*}
		where the second inequality is obtained by Lemma~\ref{lem:basic}.
	\end{proof}


	\begin{lem}\label{lem:expect2}
		The following inequality holds:
		\[
		\E\Bigg[\sum_{t=1}^T\E_t[\lb_t^\top\ib_{\Xt}-\lbh_t^\top\xbt_t]\Bigg]
		\le 2\sum_{t=1}^T\gamma_t.
		\]
	\end{lem}

	\begin{proof}
		Since $|\lb_t^\top\ib_X|\le 1$ holds for any $X\in\Fl$, we have
		\begin{align*}
		\E_t[\lb_t^\top\ib_{\Xt}-\lbh_t^\top\xbt_t]
		&=\E_t[\lb_t^\top\ib_{\Xt}]-\lb_t^\top\xbt_t
		=\sum_{X\in\Fl} p_t(X) \lb_t^\top\ib_{X} - \sum_{X\in\Fl} q_t(X) \lb_t^\top\ib_{X} \\
		&=\sum_{X\in\Fl} \gamma_t\mu \lb_t^\top\ib_{X}
		-
		\sum_{X\in\Fl} \gamma_tq_t(X) \lb_t^\top\ib_{X}
		\le 2\gamma_t.
		\end{align*}
		Summing up both sides for $t=1,\dots,T$ and taking the expectation, we obtain the claim.
	\end{proof}

    We now prove Theorem~\ref{thm:expectdetail} as
    follows.
	\begin{proof}[Proof of Theorem~\ref{thm:expectdetail}]
		Recall that we have
		$\gamma_t=\frac{t^{-1/2}}{2}$ and
		$\eta_t=\frac{\lambda}{\bdx^2}\gamma_t=
		\frac{\lambda t^{-1/2}}{2\bdx^2}$.
		The proof is readily obtained by Eq.~\eqref{ineq:expectterms}, Lemma~\ref{lem:expect1} and Lemma~\ref{lem:expect2}
		as follows:
		\begin{align*}
		\E\Bigg[\sum_{t=1}^T (\lb_t^\top \ib_\Xt - \lb_t^\top \ib_X)\Bigg]
		&=\E\Bigg[\sum_{t=1}^T \E_t[\lb_t^\top \ib_\Xt-\lbh_t^\top \xbt_t]\Bigg]
		+ \E\Bigg[\sum_{t=1}^T \E_t[\lbh_t^\top \xbt_t-\lbh_t^\top \ib_X]\Bigg] \\
		&\le
		\frac{\ln K}{\eta_T}
		+
		(e-2)d\sum_{t=1}^T
		\frac{\eta_t}{1-\gamma_t}
		+
		2\sum_{t=1}^T\gamma_t\\
		&\le
		\frac{2\bdx^2\ln K}{\lambda}\sqrt{T}
		+
		\frac{(e-2)d\lambda}{\bdx^2}\bigg(\sqrt{T}+\frac{1}{2}\ln(2\sqrt{T}-1)\bigg)
		+
		2\sqrt{T}-1.
		\end{align*}
	\end{proof}

\end{document}